\PassOptionsToPackage{dvipsnames}{xcolor}
\documentclass[11pt,a4paper,UKenglish,cleveref,autoref]{article}
\usepackage[utf8]{inputenc}
\usepackage{makecell}
\usepackage{tcolorbox}
\usepackage{tikz}
\usepackage{xcolor}
\usepackage{amsmath}
\usepackage{amssymb}
\usepackage{amsthm}
\usepackage{thmtools} 
\usepackage{thm-restate}
\usepackage{enumerate}
\usepackage{cleveref}
\usepackage{caption}
\usepackage{subcaption}
\usepackage[margin=1in]{geometry}
\newcommand{\pathonecolor}{Blue!40}
\newcommand{\pathtwocolor}{Red!60}
\newcommand{\paththreecolor}{Yellow!60}
\newcommand{\pathfourcolor}{Brown!50}
\newcommand{\pathfivecolor}{Purple!60} 
\newcommand{\pathsixcolor}{ForestGreen!80}
\newcommand{\flag}[2][1.5cm]{%
  \includegraphics[width=#1]{#2}%
}

\newtheorem{lemma}{Lemma}
\newtheorem{theorem}{Theorem}
\crefname{lemma}{Lemma}{Lemmata}

\tikzset{every picture/.append style={semithick, >=stealth}}
\tikzset{every node/.style={circle, inner sep = 2pt}}
\pgfdeclarelayer{background}
\pgfsetlayers{background,main}
\usetikzlibrary{decorations.pathreplacing, fit}

\newcommand{\paths}{\mathcal{P}}
\newcommand{\antichains}{\mathcal{A}}
\newcommand{\chains}{\mathcal{C}}
\newcommand{\ints}{\mathbb{Z}}
\newcommand{\sets}{\mathcal{S}}

\title{Maximum Coverage $k$-Antichains and Chains: A Greedy Approach} %TODO Please add

\author{Manuel Cáceres\thanks{Department of Computer Science, Aalto University, Finland and Department of Mathematics and Computer Science, University of Southern Denmark, Denmark. Supported by the Helsinki Institute for Information Technology HIIT.} \qquad Andreas Grigorjew \thanks{LAMSADE, CNRS UMR7243, Université Paris Dauphine-PSL, 75775 Paris, France and Institute of Informatics, University of Warsaw, Poland. Supported by ANR project ANR-21-CE48-0022 (S-EX-AP-PE-AL).} \qquad Wanchote Po Jiamjitrak\thanks{Department of Computer Science, University of Helsinki, Finland.} \\
Alexandru I. Tomescu\thanks{Department of Computer Science, University of Helsinki, Finland. Co-funded by the European Union (ERC, SCALEBIO, 101169716). Views and opinions expressed are however those of the author(s) only and do not necessarily reflect those of the European Union or the European Research Council. Neither the European Union nor the granting authority can be held responsible for them. Co-funded also by the Research Council of Finland grants No.~346968, 358744.\flag{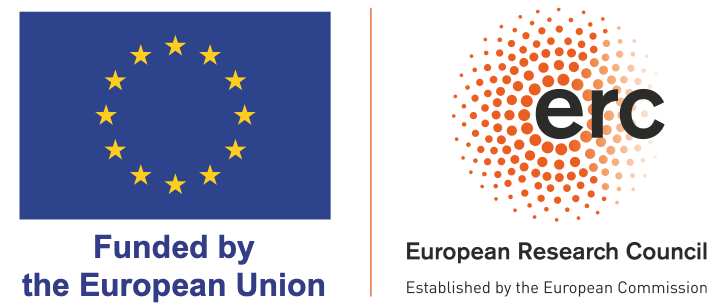}}}

\date{}

\begin{document}

\maketitle

\begin{abstract}
Given an acyclic digraph $G = (V,E)$ and a positive integer $k$, the problem of Maximum Coverage $k$-Antichains (resp. Chains) denoted as MA-$k$ (resp. MC-$k$) asks to find $k$ sets of pairwise unreachable vertices, known as antichains (resp. $k$ subsequences of paths, known as chains), maximizing the number $\alpha_k$ (resp. $\beta_k$) of vertices covered by these antichains (resp. chains). While MC-$k$ was solved in almost optimal $|E|^{1+o(1)}$ time~[Kogan and Parter, ICALP'22], the fastest algorithms for MA-$k$ are a $(k|E|)^{1+o(1)}$-time solution and a $|E|^{1+o(1)}$-time $1/2$ approximation~[Kogan and Parter, ESA'24].

We simplify and improve previous results. Specifically, we obtain the following for MA-$k$:
\begin{itemize}
    \item An algorithm running in $|E|^{1+o(1)}$ time, and an algorithm running in \emph{parameterized near-linear} $\tilde{O}(\alpha_k 
    |E|)$ time. Our algorithms are simple solutions exploiting a paths-based proof of the Greene-Kleitman theorems leveraged by the \texttt{greedy} algorithm for set cover as well as recent advances in fast algorithms for flows and shortest paths.
    \item An approximation algorithm running in \emph{parameterized linear} time $O(\alpha_1^2|V| + (\alpha_1+k)|E|)$ with approximation ratio of $(1-1/e) > 0.63 > 1/2$, beating the state-of-the-art $1/2$ approximation. Our solution uses \texttt{greedy} for antichains and a simple strategy to amortize the cost of computing consecutive maximum antichains.
\end{itemize}
Additionally, we obtain analogous results for MC-$k$ as well as the corresponding dual problems derived from the Greene-Kleitman theorems, which might be of independent interest.

We complement these results with two examples (one for chains and one for antichains) showing that, for every $k \ge 2$, \texttt{greedy} misses the tight $1/e$ portion of the optimal coverage for chains, and a $1/4$ portion for antichains. We also show that \texttt{greedy} is a $\Omega(\log{|V|})$ factor away from minimality when required to cover all vertices: previously unknown for sets of chains or antichains.
\end{abstract}

\section{Introduction}

For a directed acyclic graph (DAG) $G = (V, E)$, a chain is a subsequence of the vertices of a path (equivalently a path in the transitive closure $TC(G)$) and an antichain is a set of vertices that are pairwise unreachable (equivalently, an independent set in $TC(G)$). The problem of Maximum Coverage $k$-antichains (MA-$k$) asks to find a set of $k$ pairwise disjoint antichains whose union size is maximized: we say that one such optimal solution \emph{covers} the $\alpha_k$ vertices in this union. The problem of Maximum Coverage $k$-chains (MC-$k$) is analogously defined and every optimal solution covers $\beta_k$ vertices.

Strongly related to these coverage problems are the minimum partition problems: find the minimum number of pairwise disjoint antichains (analogously chains) partitioning the vertices of $G$, MAP (resp. MCP). Mirsky~\cite{mirsky1971dual} and Dilworth~\cite{dilworth1987decomposition} were the first to draw a connection between the maximum coverage and minimum partition problems. On the one hand, Mirsky proved that the number of antichains in a MAP equals $\beta_1$, also known as the \emph{height} of $G$~\cite{mirsky1971dual}. On the other hand, Dilworth proved that the number of chains in a MCP equals $\alpha_1$, also known as the \emph{width} of $G$~\cite{dilworth1987decomposition}. Later, Greene and Kleitman~\cite{greene1976structure,greene1976some} completed this connection with the celebrated Greene-Kleitman (GK) theorems\footnote{The generalization of Mirsky's results is due to Greene only~\cite{greene1976some}, we use GK to refer to both for simplicity.}: $\alpha_k$ (resp. $\beta_k$) is equal to the minimum \emph{$k$-norm} of a chain (resp. antichain) partition\footnote{We provide a formal definition of $k$-norm in \Cref{sec:preliminaries}.}. We denote partitions with minimum $k$-norm by MCP-$k$ and MAP-$k$, respectively.

From an algorithmic perspective, on the one hand, it is a well-known result that we can solve MAP and MC-$1$ in linear time: by a simple DP one can compute the length of the longest path ending at each vertex (solving MC-$1$), moreover a partition of the vertices according to these lengths gives a MAP. On the other hand, Fulkerson~\cite{fulkerson1956note} was the first to show a polynomial-time algorithm for MCP and MA-$1$ by reducing the problem to maximum bipartite matching. However, Fulkerson's algorithm computes a maximum matching over $TC(G)$. Modern solutions avoid working directly on $TC(G)$ since its size can be quadratic in $G$. The state-of-the-art solutions for MCP and MA-$1$ are the almost linear $|E|^{1+o(1)}$-time solution of Cáceres~\cite{caceres2023minimum}, the parameterized linear $O(\alpha_1^2|V|+|E|)$-time algorithm of Cáceres et~al.~\cite{caceres2022sparsifying,caceres2022minimum} and the parameterized near-linear $\tilde{O}(\alpha_1|E|)$-time of Mäkinen et~al.~\cite{makinen2019sparse}. The latter computes a \texttt{greedy} chain partition, which is then shrunk to minimum size. In fact, this \texttt{greedy} solution is exactly the output of the well-known \texttt{greedy} set cover~\cite{chvatal1979greedy} with sets corresponding to the chains of the DAG and thus it inherits the $O(\log{n})$ approximation ratio. This idea was originally proposed by Felsner et~al.~\cite{felsner2003recognition} for transitive DAGs, and it was later optimized for general DAGs~\cite{kowaluk2008path,makinen2019sparse}, avoiding the computation of $TC(G)$. 

For larger $k$, MC-$k$~\cite{kogan2022beating,caceres2023minimum} has been solved in almost optimal $|E|^{1+o(1)}$ time via a reduction to Minimum Cost Flow. However, until recently, the fastest solution for MA-$k$ was an $O(|V|^3)$-time algorithm from the 80's by Gavril~\cite{gavril1987algorithms}. The big difference in the running times between both problems was tantalizing and Kogan and Parter~\cite{kogan2024algorithmic} tackled this gap by exploiting algorithmic properties of the GK theorems. They present an elegant $1/2$ approximation running in time $|E|^{1+o(1)}$, as well as an approximation scheme (only for transitive DAGs) running in time $O(\alpha_k\sqrt{|E|}|V|^{o(1)}/\epsilon + |V|)$. The same paper also shows an exact solution running in $(k|E|)^{1+o(1)}$ time beating the $O(|V|^3)$-time algorithm of Gavril.

The original proof of the GK theorems by Greene and Kleitman~\cite{greene1976structure,greene1976some} was based on lattice theoretic methods, with little algorithmic insight. Saks~\cite{saks1979short} presented an elegant combinatorial proof by using the product between a poset and a chain of length $k$, which is exploited in the state-of-the-art $(k|E|)^{1+o(1)}$-time solution~\cite{kogan2024algorithmic} for MA-$k$. Later, Frank~\cite{frank1980chain} presented an algorithmic proof based on a reduction to minimum cost flow on the same graph used in Fulkerson's proof~\cite{fulkerson1956note}. As such, the corresponding algorithm runs on $TC(G)$, making it too slow. However, the structural insights discovered in this proof are exploited in the state-of-the-art $|E|^{1+o(1)}$-time $1/2$ approximation~\cite{kogan2024algorithmic}.

\subsection{Overview of our results}

In this paper we leverage a less known algorithmic proof of the GK theorems found in Schrijver's book~\cite[Theorem 14.8 \& Theorem 14.10]{schrijver2003combinatorial}. These proofs are also based on reductions to minimum cost flow that work directly on $G$, avoiding the computation of $TC(G)$. 
We argue that utilizing adaptations of Schrijver's proofs and \texttt{greedy} strategies provides a computationally fast and simple framework to solve MCP-$k$ MA-$k$, MC-$k$, and MAP-$k$, by reducing them to a constant number of minimum cost flow computations. Our first results are then obtained by using the recent results on minimum cost flow~\cite{chen2022maximum,van2023deterministic}.

\begin{restatable}{theorem}{almostlinearalgo}\label{thm:almostLinear}
    Given a DAG $G = (V, E)$ and a positive integer $k$, we can solve the problems MCP-$k$, MA-$k$, MC-$k$ and MAP-$k$ in almost optimal $|E|^{1+o(1)}$ time.
\end{restatable}

In~\cite{kogan2024algorithmic} it was shown how to compute MCP-$k$~\cite[Lemma 4.4]{kogan2024algorithmic} and MC-$k$~\cite[Lemma 4.6]{kogan2024algorithmic} by a logarithmic number of minimum cost flow computations. The corresponding results in \Cref{thm:almostLinear} are simpler and stronger since only one application of minimum cost flow is required, and we additionally solve MA-$k$ in almost optimal time, previously unknown.

Although \Cref{thm:almostLinear} solves the problems in optimal time up to subpolynomial factors, it is important to understand whether faster and/or simpler solutions can be achieved by means of parameterization. As mentioned earlier, two state-of-the-art solutions for $k=1$ (MCP and MA-$1$) run in time $O(\alpha_1^2|V|+|E|)$~\cite{caceres2022sparsifying,caceres2022minimum} and $\tilde{O}(\alpha_1|E|)$~\cite{makinen2019sparse}, which is linear and near-linear for constant values of $\alpha_1$, respectively. It is natural to ask whether such results can be obtained for general $k$. We obtain the following result.

\begin{restatable}{theorem}{nearlinearalgo}\label{thm:nearLinear}
    Given a DAG $G = (V, E)$ and a positive integer $k$, we can solve the problems MCP-$k$ and MA-$k$ in parameterized near linear $\tilde{O}(\alpha_k|E|)$ time and the problems MAP-$k$ and MC-$k$ in parameterized near linear $\tilde{O}(\beta_k|E|)$ time.
\end{restatable}

Note that these algorithms are asymptotically faster than the ones in \Cref{thm:almostLinear} when $\alpha_k, \beta_k = O(\log^c{|E|})$, running in near optimal time. To achieve this, we combine the recent results for negative cost cycles~\cite{bernstein2022negative,bringmann2023negative,haeupler2025deterministic,li2025deterministic} with the simple cycle-canceling method for minimum cost flow. In fact, the minimum cost flow values in Schrijver's reductions are $\alpha_k-|V|$ and $-\beta_k$, respectively. This observation directly derives the result for $\beta_k$. To derive the result for $\alpha_k$, we devise a $\ln{|V|}$ approximation based on \texttt{greedy} for \emph{weighted} set cover (\Cref{lem:generalizedGreedy}) and prove that the cost of such solution is $\tilde{O}(\alpha_k)$ away from the optimal. As such, this can be seen as a generalization of the simple result in~\cite{makinen2019sparse}, where only the \emph{unweighted} version of \texttt{greedy} was used. Analogous to~\cite{makinen2019sparse}, we show how to implement the required \texttt{greedy} fast.

Our last results are approximation algorithms based on (unweighted) \texttt{greedy} set cover and are independent of Schrijver's reductions. In more detail, it follows from bounds on set cover~\cite{hochbaum1996approximating} that the first $k$ \texttt{greedy} chains and antichains cover at least a $(1-1/e)$ fraction of $\beta_k$ and $\alpha_k$, respectively. This observation already beats the state-of-the-art approximation ratio of $1/2$~\cite{kogan2024algorithmic}. We show how to implement \texttt{greedy} on the set of antichains efficiently.

\begin{restatable}{theorem}{approxalgo}\label{thm:approxAlgo}
    Given a DAG $G = (V, E)$ and a positive integer $k$, there exist $(1-1/e)$-approximation algorithms solving MC-$k$ in $O(k|E|)$ time and MA-$k$ in $O(\alpha_1^2|V|+ (\alpha_1+k)|E|)$ time.
\end{restatable}

Note that these running times do not depend on $\alpha_k, \beta_k$ nor have polylog factors as in \Cref{thm:nearLinear}. Moreover, for constant value of the parameters the algorithms run in optimal time.
The result for MC-$k$ is a direct consequence of the algorithm of Mäkinen et~al.~\cite{makinen2019sparse}. For MA-$k$ we show how to compute $k$ \texttt{greedy} antichains in $O((\alpha_1+k)|E|)$ time after using the algorithm of Cáceres et~al.~\cite{caceres2022sparsifying,caceres2022minimum} to compute the first such antichain. We amortize the computation of the antichains by using the duality between MCP and MA-$1$ over a subset of vertices. As a by-product, we also obtain a $\ln{|V|}$ approximation for MAP-$k$ (\Cref{lem:generalizedGreedyAntichains}).

We highlight that the algorithmic results in \Cref{thm:nearLinear,thm:approxAlgo} belong to the line of research ``FPT inside P''~\cite{giannopoulou2017polynomial,fomin2018fully,abboud2016approximation,caceres2021a, caceres2022parameterized,himmel2024fast} of finding natural parameterizations for problems already in P.
We complement our algorithmic contributions with the first (to the best of our knowledge) study of the limitations of \texttt{greedy} applied to sets of chains or sets of antichains, as previously exploited in the literature and in this paper. As previously stated, \texttt{greedy} inherits the bounds of set cover. Namely, it is a $\ln{n}$ approximation for the partition problems~\cite{chvatal1979greedy} and a $(1-1/e)$ approximation for the coverage problems~\cite{hochbaum1996approximating}. It is known that these bound are tight for general set systems~\cite{slavik1996tight,feige1998threshold}, however, it remains unknown whether these bounds are tight for sets of chains or antichains of a DAG. We show the following.

\begin{restatable}{theorem}{upperboundcoverage}\label{thm:upperBoundCoverage}
    For every $k\ge 2$, there exists a DAG, such that $k$ \texttt{greedy} chains cover at most $(1-1/e)\beta_k$ vertices. In the case of antichains, $k$ \texttt{greedy} antichains cover at most $(3/4)\alpha_k$ vertices.
\end{restatable}

We obtain this result by creating DAG instances where \texttt{greedy} covers exactly a $(1-(1-1/k)^k)$ fraction of the optimal. For chains, we give a graph for every $k > 1$. For antichains, we give two graphs for $k=2,3$. We leave the open problem, whether the antichain case can also be generalized to all $k>3$. For the partition problems we obtain the following result.

\begin{restatable}{theorem}{lowerboundpartition}\label{thm:lowerBoundPartition}
    For MCP, there exists a class of DAGs of increasing size, such that the number of chains taken by \texttt{greedy} is a $\log_4(|V|)$ factor away from the optimal $\alpha_1=2$. For MAP, the same result holds for \texttt{greedy} antichains and $\beta_1=2$.
\end{restatable}

Although these results do not match the upper bound of $\ln{n}$, they are the first to demonstrate the approximation ratio to be $\Omega(\log{n})$. Moreover, since the results hold for constant ($=2$) optimal solution, it also disproves the number of \texttt{greedy} chains or antichains to be a function of the optimal solution, that is, $\alpha_1$ or $\beta_1$, respectively. To achieve this, we construct a recursive class of graphs such that, after hiding the vertices chosen by \texttt{greedy} in a graph, the resulting graph corresponds to the previous step in the class.

The rest of the paper is organized as follows. In the next section, we present the technical notation used in the paper as well as preliminary results, including Schrijver's reductions. The following four sections are dedicated to prove \Crefrange{thm:almostLinear}{thm:lowerBoundPartition}. We prove every theorem in several parts.

\section{Preliminaries and Schrijver's reductions}\label{sec:preliminaries}

We use $G = (V, E)$ and $k$ to denote our input DAG and a positive integer, respectively. We denote paths $P$ and chains $C$ (subsequences of paths) by their sequence of vertices, but we also use $P$ and $C$ to refer to the corresponding sets of vertices. As such, their \emph{lengths} are $|P|$ and $|C|$, respectively. Antichains are sets of pairwise unreachable vertices denoted by $A$. We use $\paths,\chains$ and $\antichains$ to denote collections of paths, chains and antichains, respectively, and $V(\paths), V(\chains), V(\antichains)$ to denote the corresponding vertices in these collections, e.g. $V(\chains) = \bigcup_{C\in\chains} C$ and $|V(\chains)|$ is the \emph{coverage} of $\chains$. In the collections of chains $\chains$ and antichains $\antichains$ used in this paper, the chains/antichains are pairwise disjoint. However, pairs of paths in collections of paths $\paths$ might intersect. If $v \in A$ (or $P$ or $C$) we say that $A$ \emph{covers} $v$, if additionally $A\in \antichains$ (or $\chains$ or $\paths$) we also say that $\antichains$ covers $v$. If $V = V(\paths)$ we say that $\paths$ is a path \emph{cover}. If $V = V(\chains)$ (resp. $V(\antichains)$) we say that $\chains$ is a chain (resp. antichain) \emph{partition}. The celebrated GK theorems state:

\begin{theorem}[Greene-Kleitman (1976)]\label{thm:GK}
The following two equalities hold:
\begin{align*}
    \alpha_k := \max_{\antichains: |\antichains| = k} |V(\antichains)| &= \min_{\chains:V = V(\chains)}\sum_{C \in \chains} \min(|C|, k)\\
    \beta_k := \max_{\chains: |\chains| = k} |V(\chains)| &= \min_{\antichains:V = V(\antichains)}\sum_{A \in \antichains} \min(|A|, k)
\end{align*}
\end{theorem}

For a chain (or antichain) partition $\chains$, $\sum_{C \in \chains} \min(|C|, k)$ is known as the $k$-norm of $\chains$. In addition to MC-$k$ and MA-$k$, we use MP-$k$ to refer to a collection of $k$ paths covering $\beta_k$ vertices\footnote{Note that any MC-$k$ can be converted into a MP-$k$ of the same size and vice versa.}. A path cover of minimum size is denoted MPC. We also use the same notation to refer to the corresponding problems of finding one such optimal collections. See \Cref{tab:acronyms} for a summary of the acronyms we use. Note that there might be no path cover with $k$-norm (as defined in \Cref{thm:GK}) exactly $\alpha_k$, as paths can be forced to cover vertices multiple times (see e.g.~\cite[Figure 1]{caceres2023minimum}). Schrijver~\cite{schrijver2003combinatorial} extended the definition of $k$-norm to collections of paths and (pairwise disjoint) antichains, and proved the following version of the GK theorems. The original result is on sets (paths and antichains) of edges: in \Cref{appendix:GKProof} we show that this can be easily adapted for sets of vertices.

\begin{restatable}[Adaptation of Schrijver (2003)]{theorem}{GKCollections}\label{thm:GKCollections}
The following two equalities hold:
\begin{align*}
    \alpha_k &= \min_{\paths}|V\setminus V(\paths)| + k|\paths|\\
    \beta_k &= \min_{\antichains}|V\setminus V(\antichains)| + k|\antichains|
\end{align*}
\end{restatable}

\begin{table}[t]
    \newcolumntype{L}[1]{>{\raggedright\arraybackslash}p{#1}}
    \caption{Acronyms used for the problems solved in this paper on input $G = (V,E), k$.}
    \centering
    \begin{tabular}{L{1.3cm}L{11.7cm}}
         Acronym     & Definition \\\Xhline{3\arrayrulewidth}
         MA-$k$     & Maximum Coverage $k$-antichains, i.e.~a set of $k$ antichains whose union size $\alpha_k$ is maximized  \\\hline
         MC-$k$     & Maximum Coverage $k$-chains, i.e.~a set of chains whose union size $\beta_k$ is maximized  \\\hline
         MP-$k$     & Maximum Coverage $k$-paths, i.e., a set of $k$ paths whose union size $\beta_k$ is maximized  \\\Xhline{3\arrayrulewidth}
         MAP        & A partition of $V$ into the minimum number of antichains, which equals $\beta_1$, the \emph{height} of $G$, by Mirsky's theorem~\cite{mirsky1971dual}  \\\hline
         MCP        & A partition of $V$ into the minimum number of chains, which equals $\alpha_1$, the \emph{width} of $G$, by Dilworth's theorem~\cite{dilworth1987decomposition}  \\\hline
         MPC        & A set of paths covering $V$ of minimum cardinality, which is equal to $\alpha_1$ by Dilworth's theorem~\cite{ntafos1979path}  \\\Xhline{3\arrayrulewidth}
         MAP-$k$    & An antichain partition $\antichains$ of minimum $k$-norm $\sum_{A \in \antichains} \min(|A|, k)$, equaling $\beta_k$ by Greene-Kleitman theorems~\cite{greene1976structure,greene1976some}; MAP-$1$=MAP  \\\hline
         MCP-$k$    & A chain partition $\chains$ of minimum $k$-norm $\sum_{C \in \chains} \min(|C|, k)$, equaling $\alpha_k$ by Greene-Kleitman theorems~\cite{greene1976structure,greene1976some}; MCP-$1$=MCP  \\\Xhline{2\arrayrulewidth}
         MAS-$k$    & A set $\antichains$ of antichains of minimum $Sk$-norm $|V\setminus V(\antichains)| + k|\antichains|$, equaling $\beta_k$ by Greene-Kleitman theorems~\cite{schrijver2003combinatorial}; MAS-$k$ equivalent to MAP-$k$  \\\hline
         MPS-$k$    & A set $\paths$ of paths of minimum $Sk$-norm $|V\setminus V(\paths)| + k|\paths|$, equaling $\alpha_k$ by Greene-Kleitman theorems~\cite{schrijver2003combinatorial}; MPS-$1$ equivalent to MPC  \\\Xhline{3\arrayrulewidth}
    \end{tabular} 
    \label{tab:acronyms}
\end{table}

Schrijver used reductions to minimum cost circulation to prove these equalities, we describe these reductions at the end of this section. For a collection of paths (resp. antichains) $\paths$, $|V\setminus V(\paths)| + k|\paths|$ is known as the $k$-norm of $\paths$, but we will use $Sk$-norm to avoid confusion. We use MAS-$k$ (resp. MPS-$k$) to denote a collection of antichains (resp. paths) of minimum $Sk$-norm. We define the \emph{partition completion} of a collection of chains (or antichains) $\chains$ as $\chains^V := \chains \cup \{\{v\} \mid v \in V\setminus V(\chains)\}$. 

\begin{restatable}{lemma}{Relations}\label{lem:relations}
 Let $\chains, \antichains, \paths$ be collections of chains, antichains and paths of a graph $G = (V, E)$, respectively. Then, the following statements hold.
    \begin{enumerate}
        \item If $V(\paths) = V(\chains)$, $|\chains| = |\paths|$ and $|C| \ge k$ for all $C \in \chains$, then the $Sk$-norm of $\paths$ equals the $k$-norm of $\chains^V$.
        \item If $|A| \ge k$ for all $A\in\antichains$, then the $Sk$-norm of $\antichains$ equals the $k$-norm of $\antichains^V$.
        \item If $\antichains$ is a MAS-$k$, then $\antichains^V$ is a MAP-$k$.
    \end{enumerate}
\end{restatable}
\begin{proof}
    \begin{enumerate}
        \item The $k$-norm of $\chains^V$ is 
        \begin{align*}
            \sum_{C\in\chains^V} \min(|C|,k) &=  \sum_{C\in \{\{v\} \mid v \in V\setminus V(\chains)\}} \min(|C|,k) + \sum_{C\in\chains} \min(|C|,k)\\  
            &= |V\setminus V(\chains)| + k|\chains| \\
            &= |V\setminus V(\paths)| + k|\paths|
        \end{align*}
        And thus equal to the $Sk$-norm of $\paths$.
        \item The $k$-norm of $\antichains^V$ is
        \begin{align*}
            \sum_{A\in\antichains^V} \min(|A|,k) &=  \sum_{A\in \{\{v\} \mid v \in V\setminus V(\antichains)\}} \min(|A|,k) + \sum_{A\in\antichains} \min(|A|,k)\\  
            &= |V\setminus V(\antichains)| + k|\antichains|
        \end{align*}
        \item It follows by the previous point and \Cref{thm:GK,thm:GKCollections}, by noting that $|A|\ge k$ for all $A\in \antichains$ (removing short, $|A|<k$, antichains decreases the $Sk$-norm).
    \end{enumerate}
\end{proof}

Our solutions rely on the use of \emph{minimum flows} and \emph{circulations}, here we review the basics needed to understand our results, for formal definitions we refer to~\cite{ahuja1993network, williamson2019network}. We say that a digraph $G' = (V', E')$, a demand function $\ell: E' \to \ints_{\ge 0}$, a capacity function $u: E' \to \ints_{\ge 0}$, a cost function $c: E' \to \ints$ and optionally $s\in V'$, $t\in V'$, is a \emph{flow network}. When specifying a flow network, if $\ell(e')$, $u(e')$, or $c(e')$ are not given for some $e'\in E'$, we assume they are $0,\infty,0$ ($\infty$ being a big enough number), respectively. We say that a flow (or circulation) $f: E' \to \ints_{\ge 0}$ is \emph{feasible} a flow network, if $\ell(e') \le f(e') \le u(e')$ for all $e' \in E'$. We denote the value of $f$ by $|f|$: the flow networks used in this paper are always DAGs up to one edge (going from $t$ to $s$), and thus the value of a feasible flow $f$ is well-defined. Moreover, in all the flow networks considered in this paper, $f$ can be \emph{decomposed} into exactly $|f|$ paths (removing the edge $(t,s)$ in the case of circulations), that is, $f$ \emph{represents} or \emph{encodes} a path collection $\paths_f$ with $|\paths_f| = |f|$. There exist algorithms computing $\paths_f$ from $f$ in optimal $O(|E'| + \texttt{output})$ time (see e.g~\cite{kogan2022beating,caceres2024practical}). Moreover, Cáceres~\cite{caceres2023minimum} showed how to compute a chain partition $\chains_f$ such that $|\chains_f| = |f|$ and $V(\paths_f) = V(\chains_f)$ in time $O(|E|\log{|f|})$. A \emph{minimum flow} (circulation) is a feasible flow minimizing $|f|$. We denote the cost of $f$ by $c(f) = \sum_{e'\in E'}f(e')c(e')$. A \emph{minimum cost flow} (circulation) is a feasible flow minimizing $c(f)$. The problems of minimum flow and minimum cost flow (circulation) have been solved in almost linear $|E|^{1+o(1)}$ time\footnote{An extra multiplicative factor of $\log{U}\cdot\log{C}$ appears in the running time, where $U$ and $C$ are the largest (absolute value of the) capacity and cost, respectively. In the flow networks used in this work, $U$ and $C$ are polynomial in $|E|$ and thus $\log{U}\cdot\log{C} = |E|^{o(1)}$.}~\cite{chen2022maximum,van2023deterministic}.

We denote the \emph{residual graph} of $G'$ and a feasible flow (circulation) $f$ by $G'_f = (V', E_f') $. Intuitively, for every edge $(u',v') = e' \in E'$, $E_f'$ contains $(u',v')$ with cost $c(e')$ only if $f(e') < u(e')$, and $(v', u')$ with cost $-c(e')$ only if $f(e') > \ell(e')$. 

A well-known result from the theory of network flows~\cite{williamson2019network} states that a feasible flow $f$ is a minimum cost circulation if and only if there is no negative cost cycle (the cost of a cycle is the sum its edges' costs) in $G'_f$.  In fact, a negative cost cycle in $G'_f$ can be used to obtain a flow of cost at most $c(f)-1$. The algorithms in \Cref{thm:nearLinear} implement a simple cycle-canceling approach boosted by the recent developments on finding negative cost cycles~\cite{bernstein2022negative,bringmann2023negative,haeupler2025deterministic,li2025deterministic}.

\begin{lemma}[\cite{bernstein2022negative,bringmann2023negative,haeupler2025deterministic,li2025deterministic}]\label{lem:cycleCancelingFast}
    Let $c^*$ be the minimum cost of a circulation in a flow network $G', \ell, u, c$. Given a feasible flow $f$, we can compute a minimum cost flow in time $\tilde{O}((c(f)-c^*+1)|E|)$.
\end{lemma}

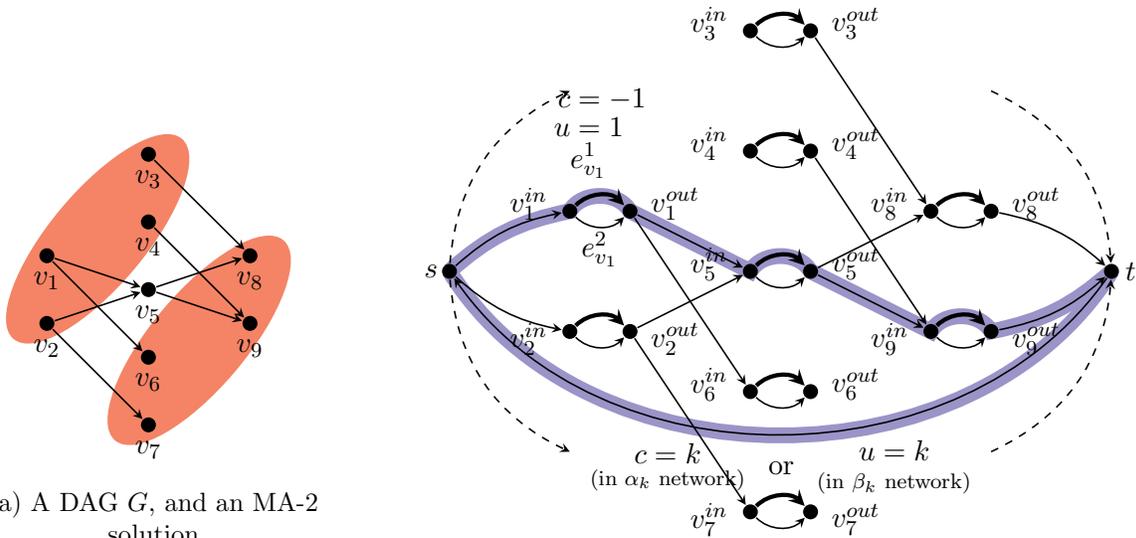
\begin{figure}[t]
\centering
    \begin{minipage}[b]{0.30\linewidth}
        \centering

        \begin{tikzpicture}[scale=0.9]  
            \draw[rotate around={-40:(0.25,1.75)}, fill=\pathtwocolor, draw=none] (0.25,1.75) ellipse (20pt and 55pt);
            \draw[rotate around={-40:(1.75,0.25)}, fill=\pathtwocolor, draw=none] (1.75,0.25) ellipse (20pt and 55pt);
            
            \node[fill=black, label={[yshift=0.1cm]below:{$v_1$}}] (A1) at (-0.5,1.5) {};
            \node[fill=black, label={[yshift=0.1cm]below:{$v_2$}}] (A2) at (-0.5,0.5) {};
            \node[fill=black, label={[yshift=0.1cm]below:{$v_3$}}] (B1) at (1,3) {};
            \node[fill=black, label={[yshift=0.1cm]below:{$v_4$}}] (B2) at (1,2) {};
            \node[fill=black, label={[yshift=0.1cm]below:{$v_5$}}] (B3) at (1,1) {};
            \node[fill=black, label={[yshift=0.1cm]below:{$v_6$}}] (B4) at (1,0) {};
            \node[fill=black, label={[yshift=0.1cm]below:{$v_7$}}] (B5) at (1,-1) {};
            \node[fill=black, label={[yshift=0.1cm]below:{$v_8$}}] (C1) at (2.5,1.5) {};
            \node[fill=black, label={[yshift=0.1cm]below:{$v_9$}}] (C2) at (2.5,0.5) {};
            \draw[->] (A1) -> (B3);
            \draw[->] (A1) -> (B4);
            \draw[->] (A2) -> (B3);
            \draw[->] (A2) -> (B5);
            \draw[->] (B1) -> (C1);
            \draw[->] (B3) -> (C1);
            \draw[->] (B2) -> (C2);
            \draw[->] (B3) -> (C2);
        \end{tikzpicture}
        \vspace{0cm} % manual hack
        \captionsetup{justification=centering}
        \subcaption{A DAG $G$, and an MA-$2$ $\{v_1,v_2,v_3,v_4\}, \{v_6,v_7,v_8,v_9\}$, in red, covering $\alpha_2 = 8$ vertices.}
        \label{subfig:G}
    \end{minipage}
    \hfill
    \begin{minipage}[b]{0.63\linewidth}
        \centering
        \begin{tikzpicture}[scale=1.6] 

            \useasboundingbox (-2.2,-2.25) rectangle (3.7,3.2);
        
            \node[fill=black, label={[yshift=0.1cm]left:{$v_1^{in}$}}] (A1in) at (-1,1.5) {};
            \node[fill=black, label={[yshift=0.1cm]right:{$v_1^{out}$}}] (A1out) at (-0.5,1.5) {};
        
            \node[fill=black, label={[yshift=-0.1cm]left:{$v_2^{in}$}}] (A2in) at (-1,0.5) {};
            \node[fill=black, label={[yshift=-0.1cm]right:{$v_2^{out}$}}] (A2out) at (-0.5,0.5) {};
        
            \node[fill=black, label={[yshift=0.1cm]left:{$v_3^{in}$}}] (B1in) at (0.5,3) {};
            \node[fill=black, label={[yshift=0.1cm]right:{$v_3^{out}$}}] (B1out) at (1,3) {};
        
            \node[fill=black, label={[yshift=0.1cm]left:{$v_4^{in}$}}] (B2in) at (0.5,2) {};
            \node[fill=black, label={[yshift=0.1cm]right:{$v_4^{out}$}}] (B2out) at (1,2) {};
        
            \node[fill=black, label={[xshift=-0.05cm]left:{$v_5^{in}$}}] (B3in) at (0.5,1) {};
            \node[fill=black, label={[yshift=0.1cm]right:{$v_5^{out}$}}] (B3out) at (1,1) {};
        
            \node[fill=black, label={[yshift=0.1cm]left:{$v_6^{in}$}}] (B4in) at (0.5,0) {};
            \node[fill=black, label={[yshift=0.1cm]right:{$v_6^{out}$}}] (B4out) at (1,0) {};    
        
            \node[fill=black, label={[yshift=-0.1cm]left:{$v_7^{in}$}}] (B5in) at (0.5,-1) {};
            \node[fill=black, label={[yshift=-0.1cm]right:{$v_7^{out}$}}] (B5out) at (1,-1) {};
            
            \node[fill=black, label={[yshift=0.1cm]left:{$v_8^{in}$}}] (C1in) at (2,1.5) {};
            \node[fill=black, label={[yshift=0.1cm]right:{$v_8^{out}$}}] (C1out) at (2.5,1.5) {};
        
            \node[fill=black, label={[yshift=-0.1cm]left:{$v_9^{in}$}}] (C2in) at (2,0.5) {};
            \node[fill=black, label={[yshift=-0.1cm]right:{$v_9^{out}$}}] (C2out) at (2.5,0.5) {};
        
            \node[fill=black, label={[xshift=0.1cm]left:{$s$}}] (s) at (-2,1) {};
            \node[fill=black, label={[xshift=-0.1cm]right:{$t$}}] (t) at (3.5,1) {};
        
            \draw[->, ultra thick] (A1in) to [bend left=45] node[midway, yshift=-0.15cm, above]{$\begin{aligned}c&=-1\\[-.2cm] u&=1\\[-.2cm]&e_{v_1}^1\end{aligned}$} (A1out); % top arc
            \draw[->] (A1in) to [bend right=45] node[midway, yshift=0.15cm, below]{$e_{v_1}^2$} (A1out); % bottom arc
            \draw[->, ultra thick] (A2in) to [bend left=45] 
            % node[midway, above, yshift=-0.4cm, xshift=0.1cm]{$\begin{aligned}c&=-1\\[-.2cm] u&=1\end{aligned}$}  
            (A2out); % top arc
            \draw[->] (A2in) to [bend right=45] (A2out); % bottom arc
            \draw[->, ultra thick] (B1in) to [bend left=45]  (B1out); % top arc
            \draw[->] (B1in) to [bend right=45] (B1out); % bottom arc
            \draw[->, ultra thick] (B2in) to [bend left=45]  (B2out); % top arc
            \draw[->] (B2in) to [bend right=45] (B2out); % bottom arc
            \draw[->, ultra thick] (B3in) to [bend left=45]  (B3out); % top arc
            \draw[->] (B3in) to [bend right=45] (B3out); % bottom arc
            \draw[->, ultra thick] (B4in) to [bend left=45]  (B4out); % top arc
            \draw[->] (B4in) to [bend right=45] (B4out); % bottom arc
            \draw[->, ultra thick] (B5in) to [bend left=45]  (B5out); % top arc
            \draw[->] (B5in) to [bend right=45] node[midway, below, yshift=1.25cm]{$\begin{array}{c}c=k\\[-0.2cm]\scriptsize{\text{(in $\alpha_k$  network)}}\end{array}$ or $\begin{array}{c}u=k\\[-0.1cm]\scriptsize{\text{(in $\beta_k$ network)}}\end{array}$} (B5out); % bottom arc
            \draw[->, ultra thick] (C1in) to [bend left=45] (C1out); % top arc
            \draw[->] (C1in) to [bend right=45] (C1out); % bottom arc
            \draw[->, ultra thick] (C2in) to [bend left=45]  (C2out); % top arc
            \draw[->] (C2in) to [bend right=45] (C2out); % bottom arc
        
            \draw[->] (A1out) -> (B3in);
            \draw[->] (A1out) -> (B4in);
            \draw[->] (A2out) -> (B3in);
            \draw[->] (A2out) -> (B5in);
            \draw[->] (B1out) -> (C1in);
            \draw[->] (B3out) -> (C1in);
            \draw[->] (B2out) -> (C2in);
            \draw[->] (B3out) -> (C2in);
        
            %\draw[->, dashed] (s) to [bend left=30] (-1,2.5);
            \draw[->] (s) to [bend left=15] (A1in);
            \draw[->] (s) to [bend right=15] (A2in);
            %\draw[->, dashed] (s) to [bend right=30] (-1,-0.5);
        
            %\draw[->, dashed] (2.5,2.5) to [bend left=30] (t);
            \draw[->] (C1out) to [bend left=15] (t);
            \draw[->] (C2out) to [bend right=15] (t);
            %\draw[->, dashed] (2.5,-0.5) to [bend right=30] (t);
            
           % \draw[->] (t) to [bend left=55] node[midway, below, yshift=2.4cm]{$\begin{array}{c}c=k\\[-0.2cm]\scriptsize{\text{(in $\alpha_k$  network)}}\end{array}$ or $\begin{array}{c}u=k\\[-0.1cm]\scriptsize{\text{(in $\beta_k$ network)}}\end{array}$} (s);

            \draw[->] (0.75,-1.75) to [bend left=45] (s);
            \draw[-] (t) to [bend left=45] (0.75,-1.75);
        
            \begin{pgfonlayer}{background}
                \draw[-,\pathonecolor, line width=6pt] (s.center) to [bend left=15] (A1in.center) to [bend left=60] (A1out.center) to (B3in.center) to [bend left=60] (B3out.center) to (C2in.center) to [bend left=60] (C2out.center) to [bend right=15] (t.center) to [bend left=45] (0.75,-1.75) to [bend left=45] (s.center);
            \end{pgfonlayer}
        
        \end{tikzpicture}
        
        \captionsetup{justification=centering}
        \subcaption{The flow network $G'$ of $G$, and an MPS-$2$ $(v_1,v_5,v_9)$, in blue, of $2$-norm equal to $6 + 2 \cdot 1 = 8 = \alpha_2$.}
        \label{subfig:Gprime}
    \end{minipage}
    \caption{A DAG $G$ and its corresponding $\alpha_k$ and $\beta_k$ networks $G'$, with their difference on edge $(t,s)$ shown. Namely, in the $\alpha_k$ network, the edge $(t,s)$ gets cost $k$, while in the $\beta_k$ network, it gets capacity $k$. For clarity, most edges of the form $(s,v^{in})$ and $(v^{out},t)$ are omitted. The edges of type $e_v^1$ are drawn thicker, and have cost $c=-1$ and capacity $u=1$.}
    \label{fig:reductions}
\end{figure}
 
We now present an adaptation of the reduction used by Schrijver in the proof of the GK theorems. For completeness, in \Cref{appendix:GKProof}, we include a complete proof of \Cref{thm:GKCollections}, adapting the proof of Schrijver to the case of sets (paths and antichains) of vertices. 

For our inputs $G = (V, E)$ and $k$, we build the graph $G' = (V', E')$ such that $V'$ contains two vertices $v^{in}$ and $v^{out}$ per vertex $v\in V$ connected by two parallel edges $e^1_v,e^2_v \in E'$ from $v^{in}$ to $v^{out}$, such that $u(e^1_v) = 1$ and $c(e^1_v) = -1$. This parallel edge construction serves the following purpose. The edge $e^1_v$ acts a ``covering edge'' that reduces the cost by one and can be used only once, and the edge $e^2_v$ allows us to use vertex $v$ in other paths but does not count towards the covering objective function. For every edge $(u,v)\in E$, $E'$ contains the edge $(u^{out}, v^{in})$. Additionally, $V'$ contains $s$ and $t$, with edges $(s, v^{in}), (v^{out}, t)\in E'$ for each $v \in V$. Finally, $E'$ contains the edge $(t, s)$: in the reduction used to prove the $\alpha_k$ part of the theorems $c(t,s)=k$ (in line with \Cref{thm:GKCollections}), in the reduction used to prove the $\beta_k$ part $u(t,s)=k$ (allowing at most $k$ paths). We call the two flow networks the \textit{$\alpha_k$ and $\beta_k$ networks}, respectively. Note that the zero circulation is a feasible circulation in both of these networks. See \Cref{fig:reductions} for an example of these networks. 
Moreover, every minimum circulation $f$ of these networks satisfies $f(e^2_v) \ge 1 \Rightarrow$ $f(e^1_v) = 1$.

In the $\beta_k$ network, this means that $c(f) = -|V(\paths_f)| = -\beta_k$, by ~\Cref{thm:GK}. Starting from the zero circulation on \Cref{lem:cycleCancelingFast}, we obtain the following Lemma.

\begin{lemma}[\Cref{thm:nearLinear}, part II]\label{lem:nearLinear2}
    Given a DAG $G = (V, E)$ and a positive integer $k$, we can solve the problems, MAP-$k$, MAS-$k$, MP-$k$ and MC-$k$ in parameterized near linear $\tilde{O}(\beta_k|E|)$ time.
\end{lemma}

A feasible circulation $f$ can be decomposed into $\paths_f$ and $\chains_f$ (to obtain MP-$k$ and MC-$k$). In \Cref{lem:antichainsExtraction} we show how to extract an MAP-$k$ and an MAS-$k$ from the minimum cost circulation. 
In contrast to the $\beta_k$ network, in the $\alpha_k$ network we have $c(f) = -|V(\paths_f)| + k|\paths_f|$ and its minimum cost is $\alpha_k-|V|$ by \Cref{thm:GKCollections}. Therefore, starting from a zero circulation, \Cref{lem:cycleCancelingFast} only produces a $\tilde{O}((\alpha_k-|V|+1)|E|)$ time algorithm. In \Cref{sec:nearLinear} we show how to efficiently compute a $\ln{|V|}$ approximation circulation for the $\alpha_k$ network. We obtain the first part of \Cref{thm:nearLinear} by starting from such a flow.

\section{MA-$k$ in Almost Linear Time}\label{sec:MA-kAlmostLinear}

The results in this section follow by connecting previous works. Nevertheless, we believe it is important to present these results, which improve over the state-of-the-art.

We first prove how to obtain \Cref{thm:almostLinear} for MPS-$k$ and MCP-$k$.

\begin{lemma}[\Cref{thm:almostLinear}, part I] \label{lem:almostLinear1}
    Given a DAG $G = (V, E)$ and a positive integer $k$, we can solve the problems MCP-$k$ and MPS-$k$ in almost optimal $|E|^{1+o(1)}+ O(\texttt{output})$ time.
\end{lemma}
\begin{proof}
    We build the $\alpha_k$ network (see \Cref{sec:preliminaries}) and compute a minimum cost circulation $f$ in $E^{1+o(1)}$ time~\cite{chen2022maximum,van2023deterministic}. Circulation $f$ can be decomposed into a collection of $|f|$ paths $\paths_f$ (for solving MPS-$k$) or $|f|$ disjoint chains $\chains_f$ such that $V(\paths_f) = V(\chains_f)$ (for solving MCP-$k$). By construction (see discussion above \Cref{lem:nearLinear2}), the cost of $f$ is $c(f) = -|V(\paths_f)| + k|\paths_f|$, since $f$ minimizes this cost it also minimizes $|V\setminus V(\paths)| + k|\paths|$, and thus $\paths_f$ is a MPS-$k$. We decompose $f$ into $\paths_f$ in $O(|E| + \texttt{output})$ time~\cite{kogan2022beating,caceres2024practical}. We can get $\chains_f$ from $f$ in $O(|E|\log{|f|}) = O(|E|\log{|V|}) = |E|^{1+o(1)}$ time~\cite{caceres2023minimum}. In~\cite[Corollary 6]{caceres2023minimum}, it is additionally stated that $\chains_f$ can be obtained from $\paths_f$ by removing the repeated vertices (up to one copy) from the paths of $\paths_f$. In other words, there is a one-to-one correspondence between $\chains_f$ and $\paths_f$ such that $C \in \chains_f$ is a subsequence of its corresponding $P \in \paths_f$. Next, we note that for every $P\in \paths_f$, $c_P := |P\setminus V(\paths_f \setminus \{P\})| \ge k$; indeed, if $c_P < k$ then the $Sk$-norm of $\paths_f \setminus \{P\}$ is $k-c_P$ units smaller than the $Sk$-norm of $\paths_f$, a contradiction. Moreover, since also $V(\chains_f) = V(\paths_f)$, the corresponding $C \in \chains_f$ to $P$ must cover at least those $c_P$ vertices, and thus for every $C \in \chains_f, |C| \ge c_P \ge k$. As such, the $k$-norm of $\chains_f^V$ is exactly the $Sk$-norm of $\paths_f$, by \Cref{lem:relations}, and thus $\chains_f^V$ is a MCP-$k$.
\end{proof}

Now we prove how to obtain \Cref{thm:almostLinear} for MP-$k$ and MC-$k$.

\begin{restatable}[\Cref{thm:almostLinear}, part II]{lemma}{AlmostLinearTwo}\label{lem:almostLinear2}
Given a DAG $G = (V, E)$ and a positive integer $k$, we can solve the problems MP-$k$, MC-$k$ in almost optimal $|E|^{1+o(1)}+ O(\texttt{output})$ time.
\end{restatable}
\begin{proof}
     We build the $\beta_k$ network (see \Cref{sec:preliminaries}) and compute a minimum cost circulation $f$ in $E^{1+o(1)}$ time~\cite{chen2022maximum,van2023deterministic}. Circulation $f$ can be decomposed into a collection of $|f|$ paths $\paths_f$ (for solving MP-$k$) or $|f|$ disjoint chains $\chains_f$ such that $V(\paths_f) = V(\chains_f)$ (for solving MC-$k$). By construction (see discussion in \Cref{sec:preliminaries}), the cost of $f$ is $c(f) = -|V(\paths_f)|$, since $f$ minimizes this cost, $\paths_f$ is a MP-$k$ and $\chains_f$ is a MC-$k$. We decompose $f$ into $\paths_f$ in $O(|E|+\texttt{output})$ time~\cite{kogan2022beating,caceres2024practical} and into $\chains_f$ in $O(|E|\log{|f|}) = O(|E|\log{|V|}) = E^{1+o(1)}$ time~\cite{caceres2023minimum}.
\end{proof}

%See \Cref{appendix:omitted} for a proof.
The proof of \Cref{thm:GKCollections} gives an explicit description of the corresponding antichain collections derived from a minimum cost circulation of the $\alpha_k$ and $\beta_k$ networks. Next, we show how to compute (from the circulation) such antichain collections efficiently.
We use this lemma to prove the final part of \Cref{thm:almostLinear}.

\begin{lemma}\label{lem:antichainsExtraction}
    Given a minimum cost circulation $f$ of the $\alpha_k$ network (resp. $\beta_k$ network). We can compute a MA-$k$ (resp. a MAS-$k$ and MAP-$k$) in time $\tilde{O}(|E|)$.
\end{lemma}
\begin{proof}
    Since there are no negative cost cycles in the residual $G'_f$, the function $d: V' \to \ints$ such that $d(v')$ is the distance of a shortest (minimum cost) path from $s$ to $v'$ in $G'_f$, is well defined. In fact, we can compute $d$ in $\tilde{O}(|E|)$ by using~\cite{bernstein2022negative,bringmann2023negative,haeupler2025deterministic,li2025deterministic}. %Alternatively, we can compute $d$ in $E^{1+o(1)}$ time via a well-known reduction to minimum cost flow: use the costs of $G'_f$ as the cost of the reduction and add a global sink $t'$ connecting every $v'\in V'$ to $t'$ with an edge with $\ell(v', t) = u(v', t) = 1$, the resulting minimum flow can be decomposed into $|V'|$ shortest paths from $s$. Instead of decomposing the paths, we can obtain $d$ in $O(|E|)$ time with a simple DFS from $s$ following non-$0$ flow edges (see e.g.~\cite{ahuja1993network}).
    The proof of \Cref{thm:GKCollections} defines $U_i :=  \{v' \in V' \mid d(v') \ge i + d(t)\}$ and $A_i = \{v\in V\mid v^{in} \in U'_i \land v^{out} \not\in U'_i\}$ for every $i\in\{1,\ldots d(s)-d(t)\}$. The proof moreover shows that $\antichains_f = \{A_1, \ldots, A_{d(s)-d(t)}\}$ is a collection of antichains and that it is a MA-$k$ in the $\alpha_k$ network and a MAS-$k$ in the $\beta_k$ network. Note that $A_i$ corresponds to the set of vertices $v\in V$ such that $d(v^{in}) \ge i + d(t)$ and $d(v^{out}) < i + d(t)$. In particular, if $d(v^{in}) > d(v^{out})$, then $v \in A_j$ for $j \in \{d(v^{out}) - d(t) + 1, \ldots, d(v^{in})-d(t)\}$. As such, we can build $\antichains_f$ in $O(|E|)$ time by checking each edge $(v^{in}, v^{out})$ and placing $v$ in one such corresponding antichain (e.g. in $A_{d(v^{in})-d(t)}$) when $d(v^{in}) > d(v^{out})$. By \Cref{lem:relations}, $\antichains^V$ is a MAP-$k$.
\end{proof}

\begin{restatable}[\Cref{thm:almostLinear}, part III]{lemma}{AlmostLinearThree}\label{lem:almostLinear3}
Given a DAG $G = (V, E)$ and a positive integer $k$, we can solve the problems MA-$k$, MAP-$k$ and MAS-$k$ in almost optimal $|E|^{1+o(1)}$ time.
\end{restatable}
\begin{proof}
    Compute a minimum cost circulation $f$ in the $\alpha_k$ network (resp. $\beta_k$ network) in $|E|^{1+o(1)}$ time~\cite{chen2022maximum,van2023deterministic}, and use \Cref{lem:antichainsExtraction} to obtain a MA-$k$ (resp. MAS-$k$ and MAP-$k$).
\end{proof}

\begin{proof}[Proof of \Cref{thm:almostLinear}]
    Combine \Cref{lem:almostLinear1,lem:almostLinear2,lem:almostLinear3}.
\end{proof}

\section{MA-$k$ in Parameterized Near-Linear Time}\label{sec:nearLinear}

We exploit the well-known \emph{greedy} algorithm for \emph{weighted set cover}~\cite{chvatal1979greedy}. In \emph{minimum weight set cover} the input is a collection of sets $\sets$ such that $\bigcup_{S\in\sets} S = V$,\footnote{The reuse of $V$ in the notation is intentional.} and an associated weight function $w: \sets \to \ints_{\ge 0}$, and the task is to find a subcollection $\sets' \subseteq \sets$ covering all elements, that is $V(\sets') = V$, and minimizing its total weight, where $w(\sets') = \sum_{S\in\sets'} w(S)$. The \texttt{greedy} algorithm maintains the set of uncovered elements $U$ and iteratively picks a set $S$ minimizing the ratio $w(S)/|U\cap S|$ until $U = \emptyset$. It is known that \texttt{greedy} is a $\ln{|V|}$ approximation~\cite{chvatal1979greedy}. We show that \texttt{greedy} is a $\ln{|V|}$ approximation for MCP-$k$ and MPS-$k$, and that it can be implemented efficiently, a generalization of the result of Mäkinen et al.~\cite[Lemma 2.1]{makinen2019sparse}.

\begin{lemma}\label{lem:generalizedGreedy}
    We can compute a $\ln{|V|}$ approximation of MCP-$k$ and MPS-$k$ in time $\tilde{O}(\alpha_k|E|/k)$.
\end{lemma}
\begin{proof}
    We define the following weighted set cover instance: the sets are the chains of $G$, and for a chain $C$ its weight is $\min{(|C|, k)}$. Note that a feasible solution for this weighted set cover instance is a chain cover (not necessarily partition), and its weight corresponds to its $k$-norm. As such, MCP-$k$'s are optimal solutions to this problem. Moreover, \texttt{greedy} is a $\ln{|V|}$ approximation~\cite{chvatal1979greedy}. We next show how to compute such a \texttt{greedy} solution efficiently. For this, recall that \texttt{greedy} at each step, chooses the chain $C$ minimizing the ratio $w(C)/|U \cap C| = \min{(|C|, k)}/|U \cap C|$, where $U$ is the set of \emph{uncovered} vertices. First, note that there is always a chain $C \subseteq U$ minimizing this ratio, and thus we assume $C \subseteq U$, and the task becomes to find a chain $C \subseteq U$ minimizing the ratio $\min{(|C|, k)}/|C|$. Note that this also ensures that the chain collection found by \texttt{greedy} is a chain partition. Second, if there is no chain $C \subseteq U$ with $|C| > k$, then the ratio is always one and any such chain minimizes the ratio, when \texttt{greedy} reaches this point it can just output every vertex in $U$ in a separated singleton path. Otherwise, \texttt{greedy} must find a chain $C\subseteq U$ minimizing $k/|C|$, but since $k$ is a constant for this minimization, it is equivalent to finding a chain $C \subseteq U$ maximizing $|C|$. Such a chain can be obtained by finding a path covering the most vertices from $U$. Mäkinen et~al.~\cite[Lemma 2.1]{makinen2019sparse} show how to compute such a path in $O(|E|)$ time. Since we stop \texttt{greedy} as soon as there is no chain $C \subseteq U$ with $|C| > k$, the running time of our approach is the number of chains $|C| > k$ times $O(|E|)$. Finally, since each chain $|C| > k$ contributes $k$ to the $k$-norm $O(\alpha_k \log{|V|})$ of the output cover, the number of such chains is $O(\alpha_k \log{|V|}/k)$. Note that the $Sk$-norm of the set of paths and the $k$-norm of the chain partition output by \texttt{greedy} is the same, by \Cref{lem:relations}.
\end{proof}

%This result suffices to show the first part of \Cref{thm:nearLinear}.

\begin{restatable}[\Cref{thm:nearLinear}, part I]{lemma}{NearLinearOne}\label{lem:nearLinear1}
Given a DAG $G = (V, E)$ and a positive integer $k$, we can solve the problems, MCP-$k$, MPS-$k$ and MA-$k$ in parameterized near linear $\tilde{O}(\alpha_k|E|)$ time.
\end{restatable}
\begin{proof}
 We use \Cref{lem:generalizedGreedy} to compute a collection of paths $\paths$ whose $Sk$-norm is at most $\alpha_k \ln{|V|}$ in $\tilde{O}(\alpha_k|E|)$ time. These paths correspond to a circulation $f$ in the $\alpha_k$ network. The cost $f$ of this circulation is then $c(f) = -|V(\paths)| + k|\paths| \le \alpha_k\ln{|V|} - |V|$, and since the minimum cost of a circulation is $\alpha_k - |V|$ we can use \Cref{lem:cycleCancelingFast} to obtain a minimum cost circulation in $\tilde{O}(\alpha_k |E|)$ time. As in the proof of \Cref{lem:almostLinear2}, we can obtain a MCP-$k$ and a MPS-$k$ from the minimum cost flow circulation (\texttt{output} size for MPS-$k$ is $O(\alpha_k|E|)$). We use \Cref{lem:antichainsExtraction} to obtain a MA-$k$.
\end{proof}

\begin{proof}[Proof of \Cref{thm:nearLinear}]
    Combine \Cref{lem:nearLinear1,lem:nearLinear2}.
\end{proof}

\section{Approximate MA-$k$ in Parameterized Linear Time}\label{sec:approxAlgo}

In the \emph{maximum coverage $k$-sets} version of set cover, the sought solution consists of $k$ sets, covering the maximum number of elements, and in this case \texttt{greedy} is stopped after $k$ iterations. It is well-known that this version of \texttt{greedy} is a $(1-1/e)$ approximation~\cite{hochbaum1996approximating}. When applied to the collection of chains/paths or antichains of a DAG, we automatically obtain the approximation ratio claimed in \Cref{thm:approxAlgo}. Therefore, the main purpose of this section is to provide a fast implementation of \texttt{greedy} in these contexts. 

In a series of works~\cite{felsner2003recognition,kowaluk2008path,makinen2019sparse} it was shown how to efficiently compute \texttt{greedy} chains/paths in $O(|E|)$ time per chain/path, the first part of \Cref{thm:approxAlgo} follows.

\begin{lemma}[\Cref{thm:approxAlgo}, part I, ~\cite{makinen2019sparse}]\label{lem:approxAlgo1}
    Given a DAG $G = (V, E)$ and a positive integer $k$, there exist $(1-1/e)$-approximation algorithms solving MP-$k$ and MC-$k$ in $O(k|E|)$ time.
\end{lemma}

We show how to efficiently compute a maximum-sized antichain $A\subseteq U$ for some subset $U\subseteq V$. We adapt a well-known reduction from MPC to minimum flow~\cite{ntafos1979path} to work with minimum path covers only required to cover $U$ instead of the whole $V$.

\begin{lemma}\label{lem:subsetMPC}
    There exists a reduction to minimum flow for the problem of finding a minimum cardinality set of paths covering $U\subseteq V$. Moreover, we can recover a maximum-sized antichain $A\subseteq U$ from a minimum flow in this network in $O(|E|)$ time.
\end{lemma}
\begin{proof}
    For $U = V$, the well-known reduction to minimum flow corresponds to $G'$ as in Schrijver's reduction but without the edge $(t,s)$ and only one copy of edges $(v^{in}, v^{out})$ for $v \in V$ with $\ell(v^{in}, v^{out}) = 1$, all other (non-specified) values for $\ell$ and $u$ are set to default (see \Cref{sec:preliminaries}). Cáceres et~al.~\cite{caceres2022sparsifying,caceres2022minimum} show that a minimum flow $f$ in this network corresponds to a MPC of $G$ (by decomposing $f$) and a maximum-sized antichain can be obtained in $O(|E|)$ time by traversing the residual $G'_f$ from $t$: the maximum-sized antichain corresponds to the vertices $v \in V$ such that $v^{out}$ is reached by $t$ in $G'_f$ but $v^{in}$ is not. Here, we prove that we obtain the same result when only required to cover $U \subseteq V$ by setting $\ell(v^{in}, v^{out}) = 1$ only for $v \in U$. First, note that a flow in this modified network can be decomposed into a collection of paths covering $U$ and a collection of paths covering $U$ corresponds to a feasible flow in this modified network. As such, a minimum flow $f$ in this modified network can be decomposed into a minimum cardinality set of paths covering $U\subseteq V$. Moreover, consider the vertices $V_t$ reached by $t$ in the residual $G'_f$ of this modified network. We define $A \subseteq U$ as $A = \{v \in U \mid v^{out}\in V_t \land v^{in} \not\in V_t\}$. First note that, if (by contradiction) there is a vertex $u \in A$ reaching a vertex $v\in A$ in $G$, then $u^{out}$ reaches $v^{in}$ in $G'$ and also in $G'_f$ (edges have default $\infty$ capacity), a contradiction since $v^{in}\not\in V_t$. As such, $A$ is an antichain. Moreover, by definition, $V_t$ is a $ts$ cut and all edges entering $V_t$ have flow equal to their lower bound. Since there is exactly $|f|$ units of flow entering $V_t$ (by flow conservation), then $|A| = |f|$, and thus $A$ is a maximum-sized antichain $A \subseteq U$. Computing $A$ reduces to computing $V_t$, which can be done in $O(|E|)$ time.
\end{proof}

Therefore, a simple strategy to obtain the $k$ \texttt{greedy} antichains is to solve the minimum flow problem defined in \Cref{lem:subsetMPC} repeatedly, where each time $U$ is defined as the vertices still uncovered. Our final result speeds up this computation by noting that a minimum flow obtained in the $i$-th step of \texttt{greedy} can be used as an initial solution for the step $i+1$. A well-known result from the theory of network flows~\cite{williamson2019network} states that for a feasible flow $f$, $f$ is a minimum flow if and only if there is no path from $t$ to $s$ in $G'_f$. A $ts$-path in $G'_f$ is known as a \emph{decrementing path} and it is used to obtain a flow of value $\le |f|-1$. As such, starting from a feasible flow $f$, we can obtain a minimum flow $f^*$ by finding $\le|f|-|f^*|$ decrementing paths in total $O((|f|-|f^*|+1)|E|)$ time.

\begin{lemma}[\Cref{thm:approxAlgo}, part II]\label{lem:approxAlgo2}
    Given a DAG $G = (V, E)$ and a positive integer $k$, there exist $(1-1/e)$-approximation algorithm solving MA-$k$ in time $O(\alpha_1^2|V| + (\alpha_1+k)|E|)$.
\end{lemma}
\begin{proof}
    We compute a MPC in time $O(\alpha_1^2|V| + (\alpha_1+k)|E|)$~\cite{caceres2022sparsifying,caceres2022minimum}, compute the corresponding minimum flow $f_1$ in the reduction of \Cref{lem:subsetMPC} and obtain the first \texttt{greedy} antichain $A_1$ in $O(|E|)$ time. To obtain the $i$-th \texttt{greedy} antichain $A_i$, for $i > 1$, we assume that $U_{i}$ are the vertices still uncovered in iteration $i$, that is $U_{i} = V \setminus \bigcup_{j = 1}^{i-1} A_j$, and that $f_{i-1}$ is a minimum flow as in \Cref{lem:subsetMPC} for $U_{i-1}$. To compute a minimum flow $f_i$ for $U_i$ we use $f_{i-1}$ as an initial solution. Note that $f_{i-1}$ is a feasible flow in the network for $U_{i}$ since $U_{i} \subseteq U_{i-1}$. Obtaining $f_i$ from $f_{i-1}$, can be done by finding $|f_{i-1}| - |f_{i}|$ decrementing $ts$ paths in the corresponding residual networks, each in $O(|E|)$ time. Once we obtain $f_{i}$, we can obtain $A_{i}$ (and $U_{i}$) in $O(|E|)$ time by the second part of \Cref{lem:subsetMPC}. The total running time after computing the MPC is then $O\left(|E|\sum_{i = 2}^k |f_{i-1}| - |f_{i}|+1 \right) = O(|E|(k + |f_1| - |f_k|)) = O((\alpha_1+k)|E|)$.
\end{proof}

\begin{proof}[Proof of \Cref{thm:approxAlgo}]
    Combine \Cref{lem:approxAlgo1,lem:approxAlgo2}.
\end{proof}

Combining the ideas from \Cref{lem:generalizedGreedy,lem:approxAlgo2} we obtain the following.

\begin{restatable}{lemma}{GeneralizedGreedyAntichains}\label{lem:generalizedGreedyAntichains}
We compute a $\ln{|V|}$ approximation of MAS-$k$ and MAP-$k$ in $\tilde{O}((\alpha_1+\beta_k/k)|E|)$ time.
\end{restatable}

\begin{proof}
    We start by computing an MPC in $\tilde{O}(\alpha_1|E|)$ time~\cite{makinen2019sparse}, and compute \texttt{greedy} antichains as in the proof of \Cref{lem:approxAlgo2} but stopping the computation whenever the next antichain size is at most $k$ (as we did for chains/paths in the proof of \Cref{lem:generalizedGreedy}) obtaining a collection of \texttt{greedy} antichains: $\antichains$. Note that, following the analysis of \Cref{lem:approxAlgo2}, the running time of computing $\antichains$ is $O((\alpha_1+|\antichains|)|E|)$. We next prove that $\antichains$ is a $\ln{|V|}$ approximation for the MAS-$k$ problem. The bound $|\antichains| \le \beta_k\ln{|V|}/k$ follows by noting that each antichain in $\antichains$ contributes exactly $k$ units to its $Sk$-norm. The fact that $\antichains$ is a $\ln{|V|}$ approximation follows by noting, as in the proof of \Cref{lem:generalizedGreedy}, that the implementation \texttt{greedy} over the set of antichains of $G$ with weights $\min(|A|,k)$ for antichain $A$, is exactly the first antichains of (unweighted) \texttt{greedy} until their size becomes at most $k$ (for MAS-$k$) plus singleton antichains for the yet uncovered vertices (for MAP-$k$). The $Sk$-norm of the antichain collection (for MAS-$k$) is equal to the $k$-norm of the antichain partition (for MAP-$k$) by \Cref{lem:relations}.
\end{proof}

\section{Limitations of Greedy on Chains/Paths and Antichains}\label{sec:limitsGreedy}

\begin{restatable}[\Cref{thm:upperBoundCoverage}, part I]{lemma}{lemupperBoundCoverageI}
    For every $k\ge 2$, there exists a DAG, such that $k$ \texttt{greedy} chains cover at most $(1-1/e)\beta_k$ vertices in the worst case, which is tight.
    \label{lem:upperBoundCoverage-I}
\end{restatable}

\begin{figure}[h]
    \centering
    \begin{tikzpicture}[
        x=0.05cm, y=1.2cm, % Scaling: 256 units = ~12.8cm wide
        font=\sffamily\footnotesize
    ]
    
    % --- Setup Rows ---
    \def\RowLength{256}
    \foreach \r in {1,2,3,4} {
        % Draw Row Label
        \node[anchor=east] at (-5, 5-\r) {\textbf{Row \r}};
        % Draw Row Background/Axis
        \draw[gray!30, fill=gray!5] (0, 5-\r-0.4) rectangle (\RowLength, 5-\r+0.4);
        \draw[gray!50] (0, 5-\r) -- (\RowLength, 5-\r);
    }
    
    % --- Path 1 (Red, w=64) ---
    % Sequence: R1 -> R2 -> R3 -> R4
    \fill[red!40] (1, 4-0.3) rectangle (64, 4+0.3) node[midway, black] {$C_1$};
    \fill[red!40] (65, 3-0.3) rectangle (128, 3+0.3) node[midway, black] {$C_1$};
    \fill[red!40] (129, 2-0.3) rectangle (192, 2+0.3) node[midway, black] {$C_1$};
    \fill[red!40] (193, 1-0.3) rectangle (256, 1+0.3) node[midway, black] {$C_1$};
    
    % Arrows P1 (Right Edge -> Left Edge)
    \draw[->, thick, red!80!black] (64, 4) -- (65, 3);
    \draw[->, thick, red!80!black] (128, 3) -- (129, 2);
    \draw[->, thick, red!80!black] (192, 2) -- (193, 1);
    
    % --- Path 2 (Blue, w=48) ---
    % Sequence: R2 -> R3 -> R4 -> R1
    \fill[blue!40] (1, 3-0.3) rectangle (48, 3+0.3) node[midway, black] {$C_2$};
    \fill[blue!40] (49, 2-0.3) rectangle (96, 2+0.3) node[midway, black] {$C_2$};
    \fill[blue!40] (97, 1-0.3) rectangle (144, 1+0.3) node[midway, black] {$C_2$};
    \fill[blue!40] (209, 4-0.3) rectangle (256, 4+0.3) node[midway, black] {$C_2$};
    
    % Arrows P2 (Right Edge -> Left Edge)
    \draw[->, thick, blue!80!black] (48, 3) -- (49, 2);
    \draw[->, thick, blue!80!black] (96, 2) -- (97, 1);
    \draw[->, thick, blue!80!black] (144, 1) -- (209, 4); % Jump back to R1
    
    % --- Path 3 (Green, w=36) ---
    % Sequence: R3 -> R4 -> R1 -> R2
    \fill[green!50!lime] (1, 2-0.3) rectangle (36, 2+0.3) node[midway, black] {$C_3$};
    \fill[green!50!lime] (37, 1-0.3) rectangle (72, 1+0.3) node[midway, black] {$C_3$};
    \fill[green!50!lime] (137, 4-0.3) rectangle (172, 4+0.3) node[midway, black] {$C_3$};
    \fill[green!50!lime] (221, 3-0.3) rectangle (256, 3+0.3) node[midway, black] {$C_3$};
    
    % Arrows P3 (Right Edge -> Left Edge)
    \draw[->, thick, green!40!black] (36, 2) -- (37, 1);
    \draw[->, thick, green!40!black] (72, 1) -- (137, 4); % Jump to R1
    \draw[->, thick, green!40!black] (172, 4) -- (221, 3);
    
    % --- Path 4 (Orange, w=27) ---
    % Sequence: R4 -> R1 -> R2 -> R3
    \fill[orange!50] (1, 1-0.3) rectangle (27, 1+0.3) node[midway, black] {$C_4$};
    \fill[orange!50] (92, 4-0.3) rectangle (118, 4+0.3) node[midway, black] {$C_4$};
    \fill[orange!50] (167, 3-0.3) rectangle (193, 3+0.3) node[midway, black] {$C_4$};
    \fill[orange!50] (230, 2-0.3) rectangle (256, 2+0.3) node[midway, black] {$C_4$};
    
    % Arrows P4 (Right Edge -> Left Edge)
    \draw[->, thick, orange!80!black] (27, 1) -- (92, 4); % Jump to R1
    \draw[->, thick, orange!80!black] (118, 4) -- (167, 3);
    \draw[->, thick, orange!80!black] (193, 3) -- (230, 2);
    
    \end{tikzpicture}
    \caption{Visualization of the graph $G_k$ in \Cref{lem:upperBoundCoverage-I}. Colored blocks represent segments of the \texttt{greedy} paths where $C_i$ is the $i^{th}$ path that \texttt{greedy} takes. For example, using the definition in the proof of \Cref{lem:upperBoundCoverage-I}, $\Gamma_{2,3} = \{ 49,\dots,96\}$ corresponds to the blue block $C_2$ in row 3.}
    \label{fig:G_4}
\end{figure}

\begin{proof}
    We construct a class of graphs $G_k$ for $k \in \mathbb{N}$ and show that \texttt{greedy} covers at most a $1 - (1-1/k)^k$-ratio of vertices of $G_k$, while $\beta_k = |V|$.
    We introduce $k^{k+1}$ vertices that can be covered by $k$ paths of each $k^k$ vertices. This set of paths represents the optimal solution and we call the $i$-th such path the $i$-th row of the graph (see \Cref{fig:G_4}). We label the vertices accordingly: $v_i^j$ with $i \in [k], \, j\in[k^k]$ denotes the $j$-th vertex in the $i$-th row.
    The idea is that \texttt{greedy} will cover a $1/k$ fraction of the remaining uncovered vertices in each row.
    
    The set of vertices $C_i \subseteq V(G_k)$ to be covered by the $i$-th \texttt{greedy} chain ($i \in [k]$) is defined as follows: $C_i = \{ v_j^\ell \mid j \in [k],\,\ell \in \Gamma_{i,j} \}$, where the index set $\Gamma_{i,j}$ is defined as the integers inside the following interval:
    \[ \Gamma_{i,j} = \mathbb{N} \cap
        \begin{cases}
            [(j-i)(1-\frac{1}{k})^{i-1}k^{k-1}+1,\, (j-i+1)(1-\frac{1}{k})^{i-1}k^{k-1}] &\text{if } j \geq i,\\
            [(k+j-i)(1-\frac{1}{k})^{i-1}k^{k-1}+(\sum_{\ell=1}^{j}(1-\frac{1}{k})^{\ell-1})k^{k-1}+1,\\\ \quad(k+j-i+1)(1-\frac{1}{k})^{i-1}k^{k-1}+(\sum_{\ell=1}^{j}(1-\frac{1}{k})^{\ell-1})k^{k-1}] & \text{if } j < i.
        \end{cases}
    \]
    
    We add the edges to $G_k$ that connect consecutive vertices $v_j^\ell$ in $C_i$ sorted in ascending order by $\ell$. Informally, the chains are paths in $G_k$ that look like stairs, where the $i$-th \texttt{greedy} chain starts on the first vertex of row $i$ and loops from row $k$ to row $1$ if $i > 1$. Clearly, $G_k$ is a DAG, because all edges have the form $(v_{i_1}^{j_1}, v_{i_2}^{j_2})$ for some $i_1,i_2 \in [k]$ and $j_2 > j_1$. We will show the following properties:
    \begin{enumerate}
        \item The sets $C_i$ are well-defined and pairwise disjoint,
        \item The set $C_i$ is of size $(1-\frac{1}{k})^{i-1}k^k$ and covers $(1-\frac{1}{k})^{i-1}k^{k-1}$ vertices in each row,
        \item The $i$-th \texttt{greedy} path covers precisely the set $C_i$ in the worst case.
    \end{enumerate}
    Showing these three properties concludes the proof, as the number of paths chosen by \texttt{greedy} divided by the optimal solution is the following:
    \[ \frac{1}{\beta_k}\left|\bigcup_{i\in[k]} C_i\right| = \frac{1}{\beta_k}\sum_{i=1}^k |C_i| = \frac{1}{k}\sum_{i=1}^k \left(1-\frac{1}{k}\right)^{i-1} \xrightarrow{k\to\infty} \left(1-\frac{1}{e}\right).\]

    \textbf{Property 1}. To verify that the sets $C_i$ are well-defined, observe that $\Gamma_{i,j}\subseteq [k^k]$. We now verify that the sets $C_i$ are pairwise disjoint for $k \geq 2$. Consider the sets $C_i$ and $C_{i-1}$ for some $i\in\{2,\dots,k\}$ and a row $j \in [k]$. We show that the intervals in $\Gamma_{i,j}$ are disjoint by analyzing several cases:
    \begin{enumerate}
        \item If $j \geq i$, the chain of $C_i$ is to the left of the chain of $C_{i-1}$ (that is, the maximum value of $C_i$ in row $j$ is smaller than the minimum value of $C_{i-1}$ in row $j$): we have $(j - i + 1)(1-\frac{1}{k})^{i-1}k^{k-1} < (j-i+1)(1-\frac{1}{k})^{i-2}k^{k-1})+1$, since $0 < (1-\frac{1}{k}) < 1$.
        \item If $j < i - 1$, the chain of $C_i$ is also to the left of the chain of $C_{i-1}$: we have $(k+j-i+1)(1-\frac{1}{k})^{i-1}k^{k-1}+(\sum_{\ell=1}^{j}(1-\frac{1}{k})^{\ell-1})k^{k-1} < (k+j-i+1)(1-\frac{1}{k})^{i}k^{k-1}+(\sum_{\ell=1}^{j}(1-\frac{1}{k})^{\ell-1})k^{k-1} + 1$, again, since $0 < (1-\frac{1}{k}) < 1$.
        \item $C_k$ is to the right of $C_1$ on all rows $j\in[k-1]$: We must verify
        \begin{align}
            j\left(1-\frac{1}{k}\right)^{k-1}k^{k-1}+\left(\sum_{\ell=1}^{j}\left(1-\frac{1}{k}\right)^{\ell-1}\right)k^{k-1} + 1 > jk^{k-1},
            \label{eq:ck-right-of-c1}
        \end{align}
        which is true if and only if
        \[
            k^{k-1} \left[ j\left(1-\frac{1}{k}\right)^{k-1} + \sum_{\ell=1}^{j}\left(1-\frac{1}{k}\right)^{\ell-1} - j \right] + 1 > 0.
        \]
        Computing the geometric sum, we have
        \[
            k^{k-1} \left[ k\left(1-\left(1-\frac{1}{k}\right)^j\right) - j\left(1-\left(1-\frac{1}{k}\right)^{k-1}\right) \right] + 1 > 0.
        \]
        It suffices to show that the expression in the bracket is non-negative for all pairs $k > j \geq 1$. This is implied by
        \begin{align*}
            (k-1)\left(1-\left(1-\frac{1}{k}\right)^{j}\right) &\geq j\left(1-\left(1-\frac{1}{k}\right)^{k-1}\right)\\
            \iff \frac{\left(1-\left(1-\frac{1}{k}\right)^{j}\right)}{j} &\geq \frac{\left(1-\left(1-\frac{1}{k}\right)^{k-1}\right)}{k-1},
        \end{align*}
        which is true, because the function $f_k(n) = (1-(1-(1/k))^n)(1/n)$ is decreasing. Hence, \Cref{eq:ck-right-of-c1} is true.
    \end{enumerate}
    The three cases above show that $C_i$ and $C_j$ are disjoint for all $i\neq j \in [k]$.

    \textbf{Property 2}. The size of $C_i$ can be computed by summing the $k$ interval lengths of $\Gamma_{i,j}$. Note that by Property 1, the $C_i$ are pairwise disjoint and that the intervals $\Gamma_{i,j}$ of a path $C_i$ all have equal length.

    \textbf{Property 3}.
    To show that \texttt{greedy} takes the path $C_\ell$, we need to show that $C_\ell$ covers the most uncovered vertices out of any path after \texttt{greedy} took $C_1,\dots,C_{\ell-1}$. After covering the first $\ell-1$ paths and before covering path $\ell$, we say that \texttt{greedy} is at step $\ell$. Given a vertex $v_i^j$, consider the set $\mathcal{R}_{i,j,\ell}$ of vertices that are in the same row, are uncovered at step $\ell$ and come after $v_j^j$:
    \[ \mathcal{R}_{i,j,\ell} = \{ v_i^p \mid p > j,\, v_i^p \text{ uncovered at step $\ell$} \}. \]
    We call an edge $(v_{i_1}^{j_1}, v_{i_2}^{j_2})$ \textit{forward} at step $\ell$ if $|\mathcal{R}_{i_1,j_2,\ell}| > |\mathcal{R}_{i_2,j_2,\ell}|$.
    We show the following Lemma.
    \begin{lemma}
        Assume \texttt{greedy} took the paths $C_1,\dots,C_t$ at steps $1,\dots,t$. Then at step $t+1$, all edges are forward.
        \label{lem:all-edges-forward}
    \end{lemma}
    \begin{proof}
        We prove this fact by induction on $t$. All edges are forward when $t = 1$, before \texttt{greedy} chose a path.
        Assume $t > 0$, and note that edges that remain in the same row are obviously forward. Consider the edges $(v_{i_1}^{j_1}, v_{i_2}^{j_2})$ of a path $C_\ell$ whose tail is contained in a different row than its head (i.e., $j_1 \neq j_2$). 
        If $\ell = t$, then the edges are forward: by the induction hypothesis, they are forward at step $t$, and after covering $C_t$, we have $\mathcal{R}_{i_1,j_1,t} = \mathcal{R}_{i_1,j_1,t+1}$ and $\mathcal{R}_{i_2,j_2,t} = \mathcal{R}_{i_2,j_2,t+1} \setminus C_t$.
        
        Let now $\ell \neq t$. Following the path $C_\ell$, it first is to the left of $C_t$ and then to the right of $C_t$. For example, in \Cref{fig:G_4} the maximum value of $C_2$ in row $2$ is smaller than the minimum value of $C_3$ in row $3$, but it is the other way round in row $3$. This is the only edge of $C_\ell$ that needs attention, as for every other edge $(v_{i_1}^{j_1}, v_{i_2}^{j_2})$ of $C_\ell$ with $j_1 \neq j_2$, we have $|\mathcal{R}_{i_1,j_1,t}| - |\mathcal{R}_{i_2,j_2,t}| = |\mathcal{R}_{i_1,j_1,t+1}| - |\mathcal{R}_{i_2,j_2,t+1}|$.
        
        We consider the case $\ell > t$, as the other case is symmetric. The path $C_\ell$ is to the left of $C_t$ on rows $\ell,\dots,k$ and $1,\dots,t-1$. Consider the intervals $\Gamma_{\ell,t-1}$ and $\Gamma_{\ell,t}$, and note that the latter contains an additional summand $(1-\frac{1}{k})^{t-1}k^{k-1}$, which is equal to the length of $C_t$ in row $t$. Hence, this edge is also forward.
    \end{proof}
    Using \Cref{lem:all-edges-forward}, we now argue that \texttt{greedy} selects the paths $C_i$ in the worst case. Consider the first \texttt{greedy} step, i.e., before any path was chosen by \texttt{greedy}. Every row has the same length $k^k$. Using Property 2, the path $C_1$ also has length $k^k$, and has the same length $k^{k-1}$ in every row. By \Cref{lem:all-edges-forward}, no path is longer than the rows. Thus, in the worst case, \texttt{greedy} chooses $C_1$.
    Using induction, we keep this invariant for every path: in the $i$-th step, every row has the same number of uncovered vertices, the path $C_i$ contains the same number of uncovered vertices, and no path contains more uncovered vertices.
\end{proof}

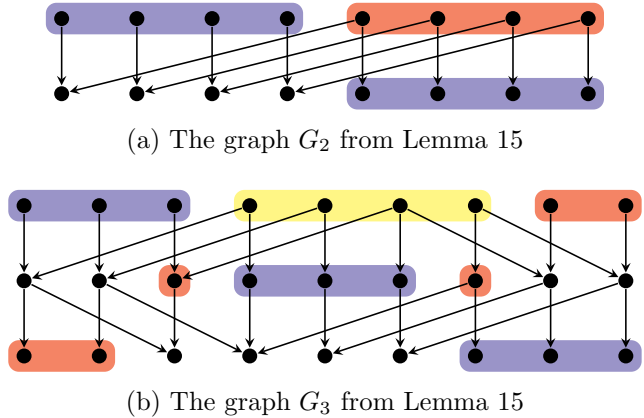
\begin{figure}
    \centering
    \begin{minipage}[b]{\linewidth}
        \centering
        \begin{tikzpicture}
        
            \foreach \i in {1,2,3,4,5,6,7,8} {
                \node[fill=black] (A\i) at (\i-1,1) {};
                \node[fill=black] (B\i) at (\i-1,0) {};
            }
            
            \foreach \i [evaluate=\i as \iplus using int(\i+1)] in {1,...,8} {   
                \draw[->] (A\i) -> (B\i);
            }
            \foreach \i [evaluate=\i as \iplus using int(\i-4)] in {5,...,8} {   
                \draw[->] (A\i) -> (B\iplus);
            }
        
            \begin{pgfonlayer}{background}
                \node[fill=\pathonecolor, draw=none, rectangle, rounded corners, fit=(A1) (A4), inner sep=1mm] (rect) {};
                \node[fill=\pathonecolor, draw=none, rectangle, rounded corners, fit=(B5) (B8), inner sep=1mm] (rect) {};
                \node[fill=\pathtwocolor, draw=none, rectangle, rounded corners, fit=(A5) (A8), inner sep=1mm] (rect) {};
            \end{pgfonlayer}
        
        \end{tikzpicture}
        \captionsetup{justification=centering}
        \subcaption{The graph $G_2$ from \Cref{lem:upperBoundCoverage-II}}
    \end{minipage}
    \\[.4cm]
    \begin{minipage}[b]{\linewidth}
        \centering
        \begin{tikzpicture}
                
            \foreach \i in {1,...,9} {
                \node[fill=black] (A\i) at (\i-1,1) {};
                \node[fill=black] (B\i) at (\i-1,0) {};
                \node[fill=black] (C\i) at (\i-1,-1) {};
            }
            
            \foreach \i [evaluate=\i as \iplus using int(\i+1)] in {1,...,9} {   
                \draw[->] (A\i) -> (B\i);
                \draw[->] (B\i) -> (C\i);
            }
            \draw[->] (A4) -> (B1);
            \draw[->] (A5) -> (B2);
            \draw[->] (A6) -> (B3);
            \draw[->] (A6) -> (B8);
            \draw[->] (A7) -> (B9);
        
            \draw[->] (B1) -> (C3);
            \draw[->] (B2) -> (C4);
            \draw[->] (B7) -> (C4);
            \draw[->] (B8) -> (C5);
            \draw[->] (B9) -> (C6);
        
            \begin{pgfonlayer}{background}
                \node[fill=\pathonecolor, draw=none, rectangle, rounded corners, fit=(A1) (A3), inner sep=1mm] (rect) {};
                \node[fill=\pathonecolor, draw=none, rectangle, rounded corners, fit=(B4) (B6), inner sep=1mm] (rect) {};
                \node[fill=\pathonecolor, draw=none, rectangle, rounded corners, fit=(C7) (C9), inner sep=1mm] (rect) {};
        
                \node[fill=\pathtwocolor, draw=none, rectangle, rounded corners, fit=(C1) (C2), inner sep=1mm] (rect) {};
                \node[fill=\pathtwocolor, draw=none, rectangle, rounded corners, fit=(B3) (B3), inner sep=1mm] (rect) {};
                \node[fill=\pathtwocolor, draw=none, rectangle, rounded corners, fit=(B7) (B7), inner sep=1mm] (rect) {};
                \node[fill=\pathtwocolor, draw=none, rectangle, rounded corners, fit=(A8) (A9), inner sep=1mm] (rect) {};
        
                \node[fill=\paththreecolor, draw=none, rectangle, rounded corners, fit=(A4) (A7), inner sep=1mm] (rect) {};
            \end{pgfonlayer}
        
        \end{tikzpicture}
        \captionsetup{justification=centering}
        \subcaption{The graph $G_3$ from \Cref{lem:upperBoundCoverage-II}}
    \end{minipage}
    \caption{Illustrations of the graphs used in \Cref{lem:upperBoundCoverage-II}. The \texttt{greedy} antichains are obtained in the order blue, red, yellow.\label{fig:upperBoundCoverage-II}}
    
\end{figure}

\begin{restatable}[\Cref{thm:upperBoundCoverage}, part II]{lemma}{UpperBoundCoverageII}\label{lem:upperBoundCoverage-II}
For every $k\ge 2$, there exists a DAG, such that $k$ \texttt{greedy} antichains cover at most $(3/4)\alpha_k$ vertices in the worst case.
\end{restatable}
\begin{proof}
    We divide this proof into two parts: first, we give two instances, for $k=2$ and $k=3$, of DAGs such that $k$ \texttt{greedy} antichains cover at most $(3/4)\alpha_k$ vertices. Secondly, we combine these two instances to create an instance for any k. Let $G_2$ be an instance for $k=2$ and $G_3$ be an instance where $k=3$. See \Cref{fig:upperBoundCoverage-II} for $G_2$ and $G_3$. In the $G_i$ instance, for $i=2,3$, $\alpha_i = |V(G_i)|$ by selecting $i$ rows. \texttt{greedy}, however, covers $8+4 =12$ vertices in $G_2$ by selecting the blue and red antichains in this order, and covers $9+6+4=19$ vertices in $G_3$ by selecting the blue, red, yellow antichains, in this order. This gives us the ratio of $12/16=3/4$ for $k=2$ and $19/27 < 3/4$ for $k=3$.

    To obtain the 3/4 bound for even $k$, we use $\frac{k}{2}$ copies of $G_2$ (called $G_2^{(1)}, G_2^{(2)}, \dots, G_2^{(\frac{k}{2})}$). Then we add edges from each vertex in $G_2^{(i)}$ to each vertex in $G_2^{(i+1)}$ for all $i$. These edges make sure that no antichain contains vertices from different $G_2^{(i)}$. \texttt{Greedy} covers $3/4$ portion of vertices of each $G_2$. For odd $k$, we use $\frac{k-1}{2}$ copies of $G_2$  (called $G_2^{(1)}, G_2^{(2)}, \dots, G_2^{(\frac{k-1}{2})}$) and one copy of $G_3$. Then we again add edges from each vertex in $G_2^{(i)}$ to each vertex in $G_2^{(i+1)}$ for all $i$, and also from each vertex in $G_2^{(\frac{k-1}{2})}$ to each vertex in $G_3$. \texttt{Greedy} covers at most $3/4$ portion of vertices of $G_3$ and each $G_2$, this gives us the final ratio of $3/4$. 
\end{proof}

\begin{proof}[Proof of \Cref{thm:upperBoundCoverage}]
    Combine \Cref{lem:upperBoundCoverage-I,lem:upperBoundCoverage-II}.
\end{proof}

We now prove the first part of \Cref{thm:lowerBoundPartition} by providing a class of graphs such that \texttt{greedy} a $\log_4(|V|)$ factor away from the optimal $\alpha_1 = 2$ for the problems MPC and MCP.
We recursively define a class of DAGs $G^C_i$ for $i\geq 2$, each of which can be covered by two paths. %We also attach weights and consider a version of \texttt{greedy} that maximizes the weight of the paths. However, this is just for visual clarity, as we can easily transform the instances to unweighted ones by replacing the vertices with paths of length equal to the weight of the vertex.

\begin{figure}
    \centering

    \begin{tikzpicture}[scale=1]
    % 1
    \node[fill=black, label=below:1] (A1) at (0,0) {};
    \node[fill=black, label=above:6] (A2) at (1,2) {};
    \node[fill=black, label=below:15] (A3) at (6,0) {};
    \node[fill=black, label=above:20] (A4) at (7,2) {};
    \node[fill=black, label=below:15] (A5) at (12,0) {};
    \node[fill=black, label=above:6] (A6) at (13,2) {};
    \draw[->] (A1) -> (A2);
    \draw[->] (A2) -> (A3);
    \draw[->] (A3) -> (A4);
    \draw[->] (A4) -> (A5);
    \draw[->] (A5) -> (A6);

    % 2
    \node[fill=black, label=below:1] (B1) at (1,0) {};
    \node[fill=black, label=above:5] (B2) at (2,2) {};
    \node[fill=black, label=below:10] (B3) at (7,0) {};
    \node[fill=black, label=above:10] (B4) at (8,2) {};
    \node[fill=black, label=below:5] (B5) at (13,0) {};
    \draw[->] (B1) -> (B2);
    \draw[->] (B2) -> (B3);
    \draw[->] (B3) -> (B4);
    \draw[->] (B4) -> (B5);

    % 3
    \node[fill=black, label=below:1] (C1) at (2,0) {};
    \node[fill=black, label=above:4] (C2) at (3,2) {};
    \node[fill=black, label=below:6] (C3) at (8,0) {};
    \node[fill=black, label=above:4] (C4) at (9,2) {};
    \draw[->] (C1) -> (C2);
    \draw[->] (C2) -> (C3);
    \draw[->] (C3) -> (C4);

    % 4
    \node[fill=black, label=below:1] (D1) at (3,0) {};
    \node[fill=black, label=above:3] (D2) at (4,2) {};
    \node[fill=black, label=below:3] (D3) at (9,0) {};
    \draw[->] (D1) -> (D2);
    \draw[->] (D2) -> (D3);

    % 5
    \node[fill, label=below:1] (E1) at (4,0) {};
    \node[fill, label=above:2] (E2) at (5,2) {};
    \draw[->] (E1) -> (E2);

    % 6
    \node[fill=black, label=below:1] (F1) at (5,0) {};

    % bottom line
    \draw[->,black] (A1) -> (B1);
    \draw[->,black] (B1) -> (C1);
    \draw[->,black] (C1) -> (D1);
    \draw[->,black] (D1) -> (E1);
    \draw[->,black] (E1) -> (F1);
    \draw[->,black] (F1) -> (A3);
    \draw[->,black] (A3) -> (B3);
    \draw[->,black] (B3) -> (C3);
    \draw[->,black] (C3) -> (D3);
    \draw[->,black] (C3) -> (D3);
    \draw[->,black] (D3) -> (A5);
    \draw[->,black] (A5) -> (B5);

    % top line
    \draw[->,black] (A2) -> (B2);
    \draw[->,black] (B2) -> (C2);
    \draw[->,black] (C2) -> (D2);
    \draw[->,black] (D2) -> (E2);
    \draw[->,black] (E2) -> (A4);
    \draw[->,black] (A4) -> (B4);
    \draw[->,black] (B4) -> (C4);
    \draw[->,black] (C4) -> (A6);

    \node[draw=none] (P6) at (0-0.3,0.55-0.3) {\textcolor{black}{$P_6$}};
    \node[draw=none] (P5) at (1-0.3,0.55-0.3) {\textcolor{black}{$P_5$}};
    \node[draw=none] (P4) at (2-0.3,0.55-0.3) {\textcolor{black}{$P_4$}};
    \node[draw=none] (P3) at (3-0.3,0.55-0.3) {\textcolor{black}{$P_3$}};
    \node[draw=none] (P2) at (4-0.3,0.55-0.3) {\textcolor{black}{$P_2$}};
    \node[draw=none] (P1) at (5-0.3,0.55-0.3) {\textcolor{black}{$P_1$}};

    \begin{pgfonlayer}{background}
        \draw[-,\pathonecolor, line width=6pt] (A1.center) -- (A2.center) -- (A3.center) -- (A4.center) -- (A5.center) -- (A6.center);
        \draw[-,\pathtwocolor, line width=6pt] (B1.center) -- (B2.center) -- (B3.center) -- (B4.center) -- (B5.center);
        \draw[-,\paththreecolor, line width=6pt] (C1.center) -- (C2.center) -- (C3.center) -- (C4.center);
        \draw[-,\pathfourcolor, line width=6pt] (D1.center) -- (D2.center) -- (D3.center);
        \draw[-,\pathfivecolor, line width=6pt] (E1.center) -- (E2.center);
        \draw[-,\pathsixcolor, line width=6pt] (5-0.17,0) -- (5+0.17,0);
    \end{pgfonlayer}

    \begin{scope}[shift={(3,-1.5)}]
        \node[fill=black,label=below:{$w$}] (n) at (0,0) {};
        \draw[->] (-0.5,0.5) to (n);
        \draw[->] (-0.5,-0.5) to (n);
        \draw[->] (n) to (0.5,0.5);
        \draw[->] (n) to (0.5,-0.5);
        \node[fill=none,draw=none] (a) at (1,0) {$\Longrightarrow$};    
        \node[fill=black] (a1) at (2,0) {};
        \node[fill=black] (a2) at (3,0) {};
        \node[fill=black] (a3) at (6,0) {};
        \node[fill=black] (a4) at (7,0) {};
        \draw[->,black] (a1) -> (a2);
        \draw[->,black] (a2) -> (4,0);
        \draw[->,black,draw=none] (a2) to node {$\cdots$} (a3) ;
        \draw[->,black] (5,0) -> (a3);
        \draw[->,black] (a3) -> (a4);
        \draw[->] (1.5,0.5) to (a1);
        \draw[->] (1.5,-0.5) to (a1);
        \draw[->] (a4) to (7.5,0.5);
        \draw[->] (a4) to (7.5,-0.5);

        % Curly brace below the nodes
        \draw [decorate,decoration={brace,amplitude=5pt,mirror,raise=.1em}] (a1.south) -- (a4.south) node[rectangle,midway,yshift=-1.2em]{$w$ vertices};
    \end{scope}

    \end{tikzpicture}
    
    \caption{Illustration of the instance $G^C_6$, with $6$ alternating paths shows as different colors. Each vertex is labeled by its weight $w$ (which can be reduced to $w$ vertices as shown in the bottom of the figure). The green path of a single vertex is $P_1$, and the blue path of length $6$ is $P_6$. The vertices are ordered by $P_6[j-1] \to P_5[j-1] \to \dots \to P_j[j-1]$ for every $j = 1,\dots,6$. Moreover, we ensure that $P_j[j-1] \to P_6[j+1]$ for all $j = 1, \dots, 4=6-2$.}
    \label{fig:lowerBoundGreedyPC}
\end{figure}
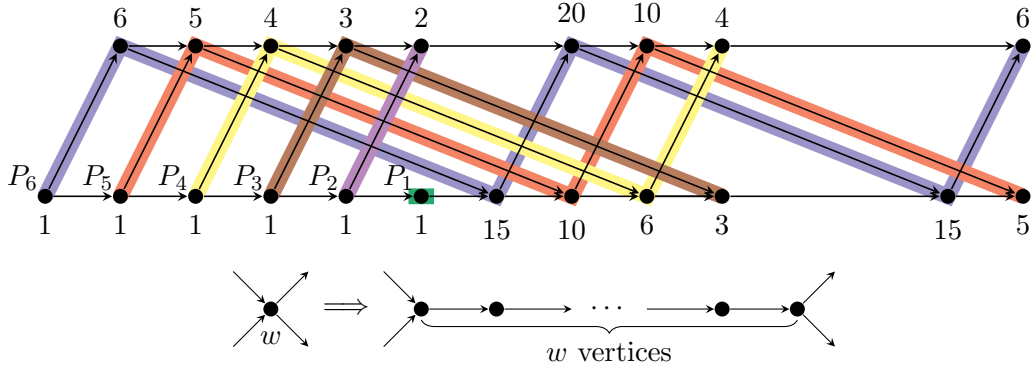
%\begin{figure}
%    \centering
%    \includegraphics[width=0.75\linewidth]{figures/greedyPathsWidth2}
%    \caption{Caption}
%    \label{fig:enter-label}
%\end{figure}
See \Cref{fig:lowerBoundGreedyPC} for an illustration of the DAGs $G^C_i = (V^C_i, E^C_i)$. We use \textit{weights $w$} as a shorthand notation for a path of $w$ vertices. Notice that the optimal solution has size 2 by taking top and bottom paths (we call these \textit{straight paths}). We then add paths with vertices that \texttt{greedy} will take (we call these \textit{alternating paths}). The graph $G^C_i$ has $i$ alternating paths $P_1, \dots, P_i$, which always start from the bottom straight path. The number of times that $P_i$ alternates between the top and bottom paths is $i-1$. We use the notation $P_i = \{ P_i[0], \dots, P_i[i] \}$ separated by each alternation.
The base case graph $G^C_1$ consists of one alternating path $P_1$ consisting of a single node with weight $1$. The $i$-th instance $G^C_i$ is constructed by adding the alternating path $P_i$ whose weights $w(P_i[j]) = \binom{i}{j}$.
The vertices of the paths are chosen to respect the following order: $P_i[j-1] \to P_{i-1}[j-1] \to \dots \to P_j[j-1]$ for all $j = 1,\dots,i$ and, in addition, $P_j[j-1] \to P_i[j+1]$ for all $j = 1,\dots,i-2$.
\begin{restatable}[\Cref{thm:lowerBoundPartition}, part I]{lemma}{LowerBoundPartitionPartI}
    For MPC and MCP, the number of chains/paths taken by \texttt{greedy} on the instances $G^C_i$ is a $\log_4(|V|)$ factor away from the optimal $\alpha_1 = 2$. \label{lem:lowerBoundPartition-partI}
\end{restatable}
\begin{proof}
    \texttt{Greedy} returns equivalent solutions for MPC and MCP. 
    We show the following:
    \begin{enumerate}
        \item In $G^C_i$, \texttt{greedy} chooses the alternating path $P_i$ in its first iteration, \label{stm:greedy_path_chooses_largest_path}
        \item $TC(G^C_i) \setminus P_i = TC(G^C_{i-1})$ for all $i > 1$. \label{stm:greedy_path_recurse}
    \end{enumerate}
    The second statement implies that after removing the vertices of $P_i$ from $TC(G^C_i)$\footnote{In the DAG $G$ vertices can be traversed by \texttt{greedy} multiple times in order to reach uncovered vertices, but in the transitive closure $TC(G)$, this is not necessary.}, the set of the remaining edges corresponds to exactly all the paths in $G^C_{i-1}$. As a result, we can recursively use Statement \ref{stm:greedy_path_chooses_largest_path} and \ref{stm:greedy_path_recurse}, to show that the paths $P_i$ are each taken by \texttt{greedy}.
    
    In the unweighted version of $G^C_i$, the number of vertices is the sum of all the weights. That is, $w(V^C_i) = |V^C_i| = \sum_{j = 1}^i |P_j|=\sum_{j = 1}^i 2^{j}-1 < 2^{i+1}$, which implies that the number of \texttt{greedy} paths is $i \geq \log_2(|V^C_i|)$. Dividing by $\alpha_1 = 2$, we obtain $\log_4(|V^C_i|)$ for the factor.

    It remains to show both statements. We first show Statement \ref{stm:greedy_path_chooses_largest_path}. Let $i \geq 2$ and let $Q = \{ Q[0], \dots, Q[|Q|-1] \}$ be the first path chosen by \texttt{greedy}. We will show that $P_i$ is the unique largest weight path in $G^C_i$ and thus, $Q = P_i$. To observe this, we first note that $P_i[0]$ is the unique source node of the DAG, which implies that $Q[0] = P_i[0]$. Assume that $Q \neq P_i$. We will locally change $Q$ until $Q = P_i$ and strictly increase its weight in the process.
    Let $j$ be the smallest index of $Q$ such that $Q[j] \neq P_i[j]$ %(we assume $|Q| \geq j+1$, because otherwise we can just extend $Q$). \andicom{I need to carefully check again if all the indices are correct :)}
    \begin{itemize}
        \item If $j = i-1$ ($P_i[j]$ last vertex of $P_i$), since $P_i$ is the longest alternating path, $|Q| = i$ and $Q[j] = P_{i-1}[j-1]$. We reroute $Q$ to $P_i[j]$, thus $w(P_{i-1}[j-1]) = \binom{i-1}{j-1} < \binom{i}{j} = w(P_i[j])$.
        \item If $j < i-1$, we reroute $Q$ by passing it through $P_i[j]$, and reconnect it again with a later vertex from $Q$ in the following way:
        \begin{itemize}
            \item If $Q$ passes through $P_i[j+1]$, it did not traverse an alternating edge, and we reroute $Q$ through $P_i[j]$ by using the alternating edge to $P_i[j+1]$. In this case, we replace \[ Q = \{ P_i[0], \dots, P_i[j-1], \mathbf{P_{i-1}[j-1], \dots, P_{j}[j-1]}, P_i[j+1], \dots \} \] by \[ Q = \{ P_i[0], \dots, P_i[j-1], \mathbf{P_i[j]}, P_i[j+1], \dots \}. \] We have \[ w(P_{i-1}[j-1]) + \dots + w(P_{j}[j-1]) = \binom{i-1}{j-1} + \dots + \binom{j}{j-1} < \binom{i}{j} = w(P_i[j]). \]
            \item If $Q$ does not pass through $P_i[j+1]$, we reroute $Q$ through $P_i[j]$ using the unique straight path to reconnect to $Q$. In this case, we replace \[ Q = \{ P_i[0], \dots, P_i[j-1], \mathbf{P_{i-1}[j-1], \dots, P_{i-\ell}[j-1]}, P_{i-\ell}[j], \dots \} \] where $\ell \leq i-j+1$, by \[ Q = \{ P_i[0], \dots, P_i[j-1], \mathbf{P_i[j], P_{i-1}[j], \dots}, P_{i-\ell}[j], \dots \}. \] We have \[ w(P_{i-1}[j-1]) + \dots + w(P_{i-\ell}[j-1]) = \binom{i-1}{j-1} + \dots + \binom{i-\ell}{j-1} < \binom{i}{j} = w(P_i[j]). \]
        \end{itemize}
    \end{itemize}
    We proceed with growing index $j$ until we reach the last vertex of $P_i$. The weight of $Q$ strictly increased after every reroute, and we thus showed that $P_i$ is the unique largest weight path in $G^C_i$. 
    Next, we show Statement \ref{stm:greedy_path_recurse}. Clearly, $TC(G^C_{i-1}) \subseteq TC(G^C_i) \setminus P_i$. To show the opposite inclusion, let $u,v \in V^C_{i-1}$, such that the edge $(u,v)$ is in $TC(G^C_i) \setminus P_i$. Let $u = P_j[\ell]$ and $v = P_k[\ell']$, where $j,k \leq i-1$. We show that there exists a $uv$-path in $G^C_{i-1}$, which shows that the edge $(u,v)$ is also in $TC(G^C_{i-1})$.
    Assume that $u$ and $v$ are covered by different straight paths (otherwise, the straight path is the $uv$-path we are looking for). Then we can consider either $u' = P_j[\ell+1]$ if $j > \ell+1$, or $u' = P_{i-1}[\ell+3]$ if $j = \ell+1$. %In both cases, $u'$ and $v$ are covered by the same straight path, and $\ell+1 \leq \ell'$ if $|P_j| > \ell+1$ or $\ell+3 \leq \ell'$ if $|P_j| = \ell+1$.
    In both cases, $u'$ is the left most vertex on the same straight path as $v$ satisfying that $(u,u')$ is an edge in $TC(G^C_{i-1})$. Thus, also $(u', v)$ is an edge in $TC(G^C_{i-1})$.
\end{proof}

    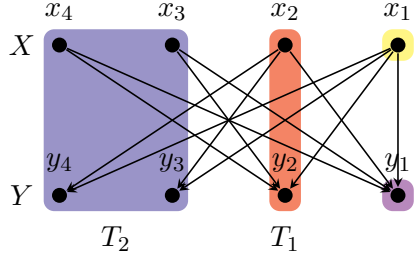
\begin{figure}[t]
        \centering
    %     \includegraphics[width=0.75\linewidth]{draftImages/lowerBoundGreedyAC.png}
    %     \caption{Caption}
    %     \label{fig:lowerBoundGreedyAC}
    % \end{figure}
    % \begin{figure}
    %     \centering
        % \includegraphics[width=0.5\linewidth]{figures/greedyAntichainsHeight2}
        \begin{tikzpicture}[scale=1]       

        % % Add a rounded corner rectangle
        % \draw[fill=\pathonecolor, rounded corners=8pt, draw=none] (1.2, -0.3) rectangle node[rectangle,above,yshift=-2cm] {$T_2$} (3.3, 2.3);
        
        % \draw[fill=\pathtwocolor, rounded corners=8pt, draw=none] (4.2, -0.3) rectangle node[rectangle,above,yshift=-2cm] {$T_1$} (4.8, 2.3);
        % \draw[fill=\paththreecolor, rounded corners=8pt, draw=none] (5.7, -0.3) rectangle (6.3, 0.3);
        % \draw[fill=\pathfivecolor, rounded corners=8pt, draw=none] (5.7, 1.7) rectangle (6.3, 2.3);
    
        \foreach \i [evaluate=\i as \ieval using 4-\i+1] in {1,2,3,4} {
            \node[fill=black,label=above:{$x_{\pgfmathtruncatemacro{\yet}{\ieval}\yet}$}] (A\i) at (1.5*\i,2) {};
            \node[fill=black,label=above:{$y_{\pgfmathtruncatemacro{\yet}{\ieval}\yet}$}] (B\i) at (1.5*\i,0) {};
        }
        \foreach \i in {1,2} {
            \draw[->] (A\i) -> (B3);
            \draw[->] (A\i) -> (B4);
            \draw[->] (A3) -> (B\i);
            \draw[->] (A4) -> (B\i);
        }
        \draw[->] (A3) -> (B4);
        \draw[->] (A4) -> (B3);
        \draw[->] (A4) -> (B4);

        \node[fill=none,draw=none] (A) at (1,2) {$X$};
        \node[fill=none,draw=none] (B) at (1,0) {$Y$};

        \begin{pgfonlayer}{background}
            \node[fill=\pathonecolor, draw=none, rectangle, rounded corners, fit=(A1) (B2), inner sep=1mm, label=below:$T_2$] (rect) {};
    
            \node[fill=\pathtwocolor, draw=none, rectangle, rounded corners, fit=(A3) (B3), inner sep=1mm, label=below:$T_1$] (rect) {};
    
            \node[fill=\paththreecolor, draw=none, rectangle, rounded corners, fit=(A4) (A4), inner sep=1mm] (rect) {};

            \node[fill=\pathfivecolor, draw=none, rectangle, rounded corners, fit=(B4) (B4), inner sep=1mm] (rect) {};
        \end{pgfonlayer}
    
        \end{tikzpicture}
        \caption{Illustration of the graph $G^A_2$ from \Cref{lem:lowerBoundPartition-partII}.}
        \label{fig:lowerBoundGreedyAC}
    \end{figure}

\begin{restatable}[\Cref{thm:lowerBoundPartition}, part II]{lemma}{LowerBoundPartitionPartII}\label{lem:lowerBoundPartition-partII}
For MAP, there exists a class of DAGs of increasing size, such that the number of antichains taken by \texttt{greedy} is a $\log_4(|V|)$ factor away from the optimal $\beta_1 = 2$ in the worst case.
\end{restatable}
\begin{proof}
    We construct a class of instances $G^A_i = (V^A_i, E^A_i)$ such that $|V^A_i| = 2 \cdot 2^i$, $\beta_1 = 2$ and \texttt{greedy} takes $i+2$ antichains. This implies the factor of $(i+2) / 2 \geq \log_4(|V^A_i|)$. See \Cref{fig:lowerBoundGreedyAC} for an illustration. The undirected graph underlying the DAG $G^A_i$ is bipartite, with $V^A_i = X \cup Y$, $X = \{ x_1, \dots, x_{2^i} \}$ and $Y = \{ y_1, \dots, y_{2^i} \}$. Since the graph is bipartite, $X$ and $Y$ are both antichains that correspond to the optimal solution. We add edges from $X$ to $Y$, such that the following sets become antichains for $1 \leq k \leq i$: $T_k = \{ x_j, y_j \mid 2^{k-1} < j \leq 2^k \}$. We let $(x_j, y_{j'}) \in E^A_i$ if and only if they are not both in the same set $T_k$ (i.e., $|T_k \cap \{ x_j, y_{j'} \}| \leq 1$ for all $1 \leq k \leq i$). Note that if a non-singleton set $U \subseteq V^A_i$ is an antichain, then it is a subset of some $T_k$, $X$ or $Y$. It can easily be verified by induction that \texttt{greedy} chooses the sets $T_k$, from $k=i$ to $1$, with two singleton vertices $x_1$ and $y_1$ left over, that are chosen as two new antichains, hence $i+2$ antichains in total.
\end{proof}

Note that \texttt{greedy}'s worst case solutions are not unique in \Cref{lem:upperBoundCoverage-I}, \Cref{lem:upperBoundCoverage-II} and \Cref{lem:lowerBoundPartition-partII}, however, they are unique in \Cref{lem:lowerBoundPartition-partI}.

\begin{proof}[Proof of \Cref{thm:lowerBoundPartition}]
    Combine \Cref{lem:lowerBoundPartition-partI,lem:lowerBoundPartition-partII}.
\end{proof}

\bibliography{references}

\appendix

\newpage

\section{Proof of the Greene-Kleitman Theorems}\label{appendix:GKProof}

\GKCollections*
\begin{proof}
This proof mimics the proofs of ~\cite[Theorem 14.8 \& Theorem 14.10]{schrijver2003combinatorial} for the case of paths/antichains of vertices. We first prove the inequality $\le$, which holds for any $k$ antichains/$k$ paths and any path/antichain collection. For the $\alpha_k$ version, let $\antichains$ be $k$ antichains, $|\antichains| = k$, and $\paths$ a collection of paths. Then, note that 
\begin{align*}
    |V(\antichains)|    &\le |V\setminus V(\paths)| + |V(\paths) \cap V(\antichains)| \\
        &\le |V\setminus V(\paths)| + \sum_{P\in \paths} \sum_{A\in\antichains}|P \cap A| \\
        &\le |V\setminus V(\paths)| + k|\paths|
\end{align*}
where the last inequality follows since a path and an antichain can intersect in at most one vertex, i.e., $|P\cap A| \le 1$. The $\beta_k$ version of the inequality follows by the same argument (by exchanging the roles of $\paths$ and $\antichains$). To prove that the values meet in the optimum we use the $\alpha_k$ and $\beta_k$ networks described in \Cref{sec:preliminaries}. Consider a minimum cost circulation $f$ in one of the networks. As argued in the main paper, if $f$ is minimum in the $\alpha_k$ network, then we can obtain a MPS-$k$ $\paths_f$ by decomposing $f$, and if $f$ is minimum in the $\beta_k$ network, then we can obtain a MP-$k$ $\paths_f$ by decomposing $f$. In both networks, consider the function $d: V' \rightarrow \ints$ such that $d(v')$ is the distance of a shortest (minimum cost) path from $s$ to $v'$ in the residual $G'_f$. First, note that since $f$ is minimum, there are not negative cost cycles in $G'_f$ and thus $d$ is well defined. Moreover, $d(s) = 0$ by definition, and $d(v') \le 0$ for all $v'\in V'$ since there is always a path $P = s, v^{in}, v^{out}, t$ (using $e_v^2$) in $G'_f$ such that all edge costs are $0$ and $v' \in P$. Furthermore, since $d$ is the shortest path distance from $s$ in $G'_f$, it follows that for all $e' = (u', v') \in E'_f$, $d(v') \le d(u') + c(e')$. In particular, if both an edge $e' = (u', v')$ and its reverse are in $E'_f$, then $d(v') = d(u') + c(e')$. For our networks, this implies:

\begin{enumerate}
    \item $d(v') \le d(u')$ if $e' = (u', v') \in E'\setminus (\{e^1_v \mid v \in V\}\cup \{(t,s)\})$, and $d(v') = d(u')$ if $f(e') \ge 1$.
    \par Cost in such edges is $c(e') = 0$ and its reverse is in $E'_f$ if and only if $f(e') \ge 1$.
    \item $d(v^{out}) \le d(v^{in}) - 1$ if $f(e^1_v) = 0$, and $d(v^{out}) \ge d(v^{in}) - 1$ if $f(e^1_v) = 1$.
    \par Cost in such edges is $c(e^1_v) = -1$ and $e^1_v \in E'_f \Leftrightarrow f(e^1_v) = 0 \Leftrightarrow$ reverse of $e^1_v \not\in E'_f$.
    \item For the $\alpha_k$ network, $d(s) = d(t) + k$, since we assume $f(t, s) \ge 1$: otherwise, there is no path of length more than $k$ (one such path would reduce $c(f)$) and by Mirsky's theorem~\cite{mirsky1971dual} a MAP is a MA-$k$ whose coverage ($\alpha_k = |V|$) matches the $Sk$-norm of $\paths_f = \emptyset$.
    \item For the $\beta_k$ network, we assume $f(t,s) = k$: otherwise we can redefine $f$ to be a different minimum cost circulation with the same cost and $f(t,s) = k$ (e.g. pick arbitrary $v\in V$ and push $k-f(t,s)$ units of flow on $(s, v^{in}), e^2_v, (v^{out}, t)$ and $(t,s)$ at total cost $0$).
\end{enumerate}

We define the following subset of vertices $U_i := \{v' \in V' \mid d(v') \ge i + d(t)\}$ for $i \in \{1, \ldots, d(s) - d(t)\}$. By definition, $U'_i$ is an $st$-cut ($s\in U'_i$, $t\not\in U'_i$). Moreover, there are no edges, other than $(t,s)$ going into $U'_i$ in $G'$: indeed, if $(u', v') \in E'\setminus \{(t,s)\}, u' \not\in U'_i, v'\in U'_i$, then by \textbf{1.} and \textbf{2.}, $d(v') \le d(u')$, but also $d(u') < i + d(t)$ and $d(v') \ge i + d(t)$ implying $d(u') < d(v')$, a contradiction. We use $U_i$ to define subsets of vertices of the input DAG $G$, $A_i = \{v\in V\mid v^{in} \in U'_i \land v^{out} \not\in U'_i\}$. In fact, $A_i$ is an antichain: if $u\neq v \in A$ are such that there is a $uv$ path in $G$, then there is a $u^{out}v^{in}$ path in $G'\setminus \{(t,s)\}$, but such a path would cross from $V'\setminus U_i$ to $U_i$, a contradiction. We define $\antichains_f = \{A_{1}, \ldots, A_{d(s)-d(t)}\}$ to be the collection of antichains extracted from $f$. Note that in the $\alpha_k$ network, by \textbf{3.}, $|\antichains_f| = k$. And in the $\beta_k$ network, by \textbf{4.}, $|\paths_f| = k$. These collections satisfy:

\begin{enumerate}[a.]
    \item $V\setminus V(\paths_f) \subseteq V(\antichains_f)$: indeed, if $v \in V\setminus V(\paths_f)$, then $f(e^1_v) = 0$ and, by \textbf{2.}, $d(v^{out}) \le d(v^{in}) - 1$. But then, $v \in A_{d(v^{in}) - d(t)} \subseteq V(\antichains_f)$. Note that, $1 \le d(v^{out}) + 1 - d(t) \le d(v^{in}) - d(t) \le d(s) - d(t)$ since both $(s, v^{in}), (v^{out}, t) \in E'_f$.
    \item Thus also $V\setminus V(\antichains_f) \subseteq V(\paths_f)$.
\end{enumerate}

For the $\alpha_k$ network we have:
\begin{align*}
    k|\paths_f| &=\\
   \text{(by \textbf{3.})}  &= (d(s) - d(t))f(t,s)\\
  \text{($f$ is circulation)}  &= \sum_{e'=(u',v')\in E'\setminus\{(t,s)\}} (d(u')-d(v'))f(e')\\
    &= \sum_{e'=(u',v')\in E'\setminus (\{e^1_v \mid v \in V\}\cup \{(t,s)\})\})} (d(u')-d(v'))f(e') \\
    &+ \sum_{e^1_v, v \in V} (d(v^{in})-d(v^{out}))f(e^1_v)\\
  \text{(by \textbf{1.})}  &= \sum_{e^1_v, v \in V} (d(v^{in})-d(v^{out}))f(e^1_v) \\
  \text{(by \textbf{2.})}& = |\{v \in V \mid f(e^1_v) = 1 \land d(v^{in})-d(v^{out}) = 1\}|   \\
  &\le |V(\paths_f) \cap V(\antichains_f)|
\end{align*}

For the $\beta_k$ network we have:
\begin{align*}
    k|\antichains_f| &=\\
    \text{(by \textbf{4.})} &= f(t,s) (d(s) - d(t)) \\
    \text{(previous derivation)}&\le   |V(\paths_f) \cap V(\antichains_f)|
\end{align*}

For $\alpha_k$ the proof concludes by:
\begin{align*}
    |V\setminus V(\paths_f)| + k |\paths_f| &=\\
    \text{(by \textbf{a.})} &= |V(\antichains_f)\setminus V(\paths_f)| + k |\paths_f| \\
    \text{(derivation for $\alpha_k$ network)} &\le |V(\antichains_f)\setminus V(\paths_f)| + |V(\paths_f) \cap V(\antichains_f)| \\
    &= |V(\antichains_f)|
\end{align*}

Analogously for $\beta_k$:
\begin{align*}
    |V\setminus V(\antichains_f)| + k |\antichains_f| &=\\
    \text{(by \textbf{b.})} &= |V(\paths_f)\setminus V(\antichains_f)| + k |\antichains_f| \\
    \text{(derivation for $\beta_k$ network)} &\le |V(\paths_f)\setminus V(\antichains_f)| + |V(\paths_f)\cap V(\antichains_f)| \\
    &= |V(\paths_f)|
\end{align*}
\end{proof}

\end{document}